\documentclass[runningheads]{llncs}
\usepackage[T1]{fontenc}

\usepackage{wrapfig}
\usepackage{booktabs}   
\usepackage{subcaption} 

\usepackage{etoolbox}
\usepackage{float}
\usepackage[section]{placeins}
\usepackage{color}
\usepackage{colortbl}
\usepackage{cancel}
\usepackage{dashbox}
\usepackage[normalem]{ulem}
\usepackage{amsmath}
\usepackage{mathtools}
\usepackage{amssymb}
\usepackage{ textcomp }
\usepackage{makecell}
\usepackage{boldline}
\usepackage{tikz}
\newcommand{\stkout}[1]{\ifmmode\text{\sout{\ensuremath{#1}}}\else\sout{#1}\fi}
\definecolor{Red}{RGB}{255,0,0}
\definecolor{DeepGreen}{RGB}{5,102,8}
\definecolor{LGray}{gray}{0.9}
\definecolor{LCyan}{rgb}{0.88,1,1}

\usepackage{wrapfig}
\usepackage[frozencache,cachedir=.]{minted}
\usepackage{bbding}
\usepackage{listings} 
\usepackage{semantic}
\usepackage{tikz,pgf}
\usepackage{stmaryrd}
\usepackage{xcolor,colortbl}
\usepackage{ebproof}
\usepackage{relsize}

\usepackage{tikz}
\usepackage{xfrac} 
\usepackage{authorcomments}
\newcommand{\name}{\textsf{Hegel}}

\newcommand{\T}[1]{\mbox{$\xhookrightarrow{#1}$}}
\newcommand{\automata}{\textit{Liquid Tree Automata}\,}
\newcommand{\automaton}{\textit{Liquid Tree Automaton}\,}
\newcommand{\calculus}{\mbox{$\lambda_{\S{lta}}$}}

\usepackage{stmaryrd}

\newcommand{\eip}{\mbox{$\mathit{e}$}}

\definecolor{mygreen}{rgb}{0,0.6,0}
\definecolor{mygray}{rgb}{0.5,0.5,0.5}
\definecolor{mymauve}{rgb}{0.58,0,0.82}
\definecolor{amber}{rgb}{1.0, 0.75, 0.0}
\definecolor{bargreen}{rgb}{0.0, 0.5, 0.0}
\renewcommand{\S}[1]{\mbox{$\mathsf{#1}$}}

\lstnewenvironment{codetable}
  {\lstset{language=ML,
    basicstyle=\footnotesize,
    }%
  }
  {}
\usepackage{caption}

\AtBeginEnvironment{minted}{\let\itshape\relax}

\tikzoption{base font}{\def\tikz@base@textfont{#1}}
\tikzset{
  base font=\sffamily,
}
\usepackage{amsmath}
\usepackage{supertabular}
\usepackage{ifthen}
\usepackage{alltt}
\usepackage{color}

\usepackage{proof}
\usepackage{multirow}
\usepackage{multicol}
\usepackage{semantic}
\usepackage{cancel}
\usepackage{algpseudocode}
\usepackage[noend,noline]{algorithm2e}
\usepackage{pgfplots, pgfplotstable}

\usepackage[normalem]{ulem}

\begin{document}


\author{Ashish Mishra\inst{1} \and Suresh Jagannathan\inst{2}}
\authorrunning{A. Mishra and S. Jagannathan}
%
\institute{IIT Hyderbad, India  \email{mishraashish@cse.iith.ac.in}\\
\and
Purdue University, USA \email{suresh@cs.purdue.edu}}

%


\renewcommand{\ADD}[1]{#1}


\title{Liquid Tree Automata}
%
%
%
%
%
\maketitle              

\begin{abstract}
Component-based synthesis (CBS) aims to generate loop-free programs
from a set of libraries whose methods are annotated with
specifications and whose output must satisfy a set of logical
constraints, expressed as a query. The effectiveness of a CBS
algorithm critically depends on the severity of the constraints
imposed by the query. The more exact these constraints are, the
sparser the space of feasible solutions.  This maxim also applies when
we enrich the expressivity of the specifications affixed to library
methods.  In both cases, search must now contend with constraints that
may only hold over a small number of the possible execution paths that
can be enumerated by a CBS procedure.

\quad\ In this paper, we address this
challenge by equipping CBS search with the ability to reason about
\emph{logical similarities} among the paths it explores. Our setting
considers library methods equipped with refinement-type specifications
that enrich ordinary base types with a set of rich logical qualifiers
to constrain the set of values accepted by that type.

\quad\ For efficient representation and enumeration of this space, we
introduce a novel tree automata variant called \automata{} (LTA) whose
construction is driven by the typing rules of a refinement type
system. This allows us to leverage subtyping constraints over the
refinement types associated with enumerated terms to enable reasoning
about similarity among candidate solutions as search proceeds, using
this notion of similarity to eagerly merge LTA states.  By doing so,
we avoid exploration of semantically similar paths, leading to a
significantly improved search procedure.  We present an implementation
of this idea in a tool called \name\, and provide a comprehensive
evaluation that demonstrates \name's ability to synthesize solutions
to complex CBS queries that go well-beyond the capabilities of the
existing state-of-the-art.



\end{abstract}

\section{Introduction}
\label{sec:introduction}

{\scriptsize
\begin{figure*}[h!]
\begin{minted}[fontsize = \footnotesize, mathescape=true, escapeinside=&&]{ocaml}
&\textcolor{blue}{(a) A part of the library}&
(*take : (x : nat) $\rightarrow$ [a] $\rightarrow$ {v : [a] $\mid$ len (v) $\leq$ x $\vee$ len (v) = 0}*)
val take : int &$\rightarrow$& [a] &$\rightarrow$& [a]

(*splitAt : (x : nat) $\rightarrow$ (xs : [a]) $\rightarrow$ {v : (f : [a], s : [a]) 
$\mid$  len (f) $\leq$ x $\wedge$ (len (s) $\leq$ len (xs) - x)}*) 
splitAt : int &$\rightarrow$& [a] &$\rightarrow$& ([a],[a])

(*decr : (x : nat) $\rightarrow$  {v:int $\mid$ x = x - 1}*)
val decr : int &$\rightarrow$& int

(*fst : (x : ([a], [a])) $\rightarrow$  {v : [a] $\mid$ v = fst (x)}*)
val fst : ([a], [a]) &$\rightarrow$& [a]

(*snd : (x : ([a], [a])) $\rightarrow$  {v : [a] $\mid$ v = snd (x)}*)
val snd : ([a], [a]) &$\rightarrow$& [a]

(*drop : (x : nat) $\rightarrow$  (xs : [a]) $\rightarrow$ {v : [a] $\mid$ len (v) $\leq$ len (xs) - x}}*)
val drop : int -> [a] &$\rightarrow$& [a]

&\textcolor{blue}{(b) A functional query type}&
(* goal : (x:nat) $\rightarrow$  (y : nat) $\rightarrow$  (z : [a]) $\rightarrow$  { v : (f : [a], s : [a]) 
          $\mid$ len (f) $\leq$ x $\wedge$ (len (s) $\leq$ len (z) - y} *)
goal : int &$\rightarrow$& int &$\rightarrow$& [a] &$\rightarrow$& ([a], [a])

&\textcolor{blue}{(c) A few correct solutions.}&
&\textcolor{blue}{uncommented (black): unrefined, commented (green): refined.}&
(* fun x y z $\rightarrow$  ((take x (fst (splitAt y z))), snd (splitAt y z) ) *)
&fun& x y z &$\rightarrow$& splitAt y (drop x z)
&fun& x y z &$\rightarrow$& splitAt x (take y z)
&$\ldots$&
\end{minted}
\caption{Motivating Synthesis Problem}
\label{fig:motivation}
\vspace{-.3in}
\end{figure*}
}
Component-based synthesis (CBS) aims to generate loop-free programs
from a library of components, typically defined as methods provided by
an API. At the heart of any CBS implementation is a search problem
over a hypothesis space of programs that “glue” components together
using basic control primitives, such as conditionals and function
applications.  If the attributes defining the behavior of components
are not overly constrained, or when queries are reasonably general,
the search for a feasible solution can be tractable. However, as
specifications become more precise, the set of feasible programs
becomes a small fraction of the overall search space, making synthesis
significantly more challenging.  Intuitively, we can define CBS search
as a reachability analysis over a graph that relates candidate methods
based on their type or other similar defining attributes. For example,
a node in this graph associated with a method that has a particular
result type can be connected to any node corresponding to a method
that accepts an argument of this type. Such connections can be used by
the synthesizer to produce a candidate solution that connect inputs to
outputs through sequences of component applications.  Prior
work~\cite{sypet,tygus,ecta} has considered the construction of such
graphs using simple type-based specifications. In this paper, we
propose to allow richer query specifications in the form of refinement
types~\cite{JV21} that both decorate library methods and serve as the
basis for synthesis queries.

Recent advances in automated theorem proving and formal verification
have made it increasingly common to have libraries be equipped with
such rich specifications~\cite{vocal,fstar}, motivating synthesis
techniques that can effectively exploit and tame such
information. While expressive specifications using refinement types
have been used previously to guide synthesis
\cite{cobalt-tech,synquid}, these approaches did not address the
unique challenges that arise in a component-based synthesis setting,
where na\"ive enumerative synthesis techniques are ineffective.


The synthesis query in Fig.~\ref{fig:motivation} illustrates this
challenge. The synthesizer is given two inputs, the first is a library
of components, each specified by a type signature enriched with
refinement pre/postconditions, as comments.
Fig.~\ref{fig:motivation}(a) shows a representative fragment of such
a library; in practice, libraries may contain hundreds of components
spanning standard data structures and operations.  The second input is
a synthesis query (called \S{goal}), also expressed as a refinement
type, specifying both the desired input–output base types and semantic
constraints on the result (Fig.~\ref{fig:motivation}(b)).  The
synthesis task is to construct a composition of library components
that satisfies the query’s refinement specification.  Fig.~\ref{fig:motivation}(c)
shows solutions for the simple (unrefined) goal in uncommented black, and for
the refined query, commented in green.
\begin{wrapfigure}{r}{.20\textwidth}
\includegraphics[scale=.45]{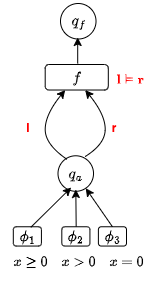}
\caption{An Example LTA} 
\label{fig:compare-single}
\vspace*{-.25in}
\end{wrapfigure}
A natural approach to handling such refined queries is to synthesize
programs using only the unrefined base types of the query and library
components~\cite{sypet,tygus,ecta} , and to subsequently apply a
verification check to discard candidates that violate refinement
constraints.  Although simple, this strategy is unlikely to be
effective at scale; not only is the solution space now very sparse,
thus making simple path exploration inefficient, but the (base) type
specifications of library methods are too weak to meaningfully guide
the synthesizer toward a correct solution.

To illustrate, we evaluated this synthesis problem (the commented
goal, plus a mid-sized component library of around 300 functions) on
two state-of-the-art systems that support such refined specifications.
Synquid~\cite{synquid} performs refinement type-driven synthesis by
enumerating well typed programs from the library, while
Hoogle+~\cite{hoogleplus} encodes refinements using input-output
examples, and does post-facto filtering. Both tools failed to
synthesize a solution for this query (e.g, the commented green
definition in Fig.~\ref{fig:motivation}(c)) within a timeout of
three minutes. The central challenge here is the ability to compactly
represent the space of well typed programs and to search this space
efficiently. Existing component based synthesis techniques struggle on
both fronts. In contrast, because our synthesis procedure, \name{},
constructs a compact representation of the well-typed search space to
enable efficient exploration, it was able to synthesize the solution
in 15 seconds.

\subsection{Solution: \automata}
Several data structures could in principle be used to represent large spaces of candidate programs, including
Version Space Algebras~\cite{vsa},
e-graphs~\cite{egg-synthesis} and Finite Tree Automata
(FTA)~\cite{tata}. In particular, FTA have been shown to be
effective in representing the space of untyped programs, satisfying a set of input-output
examples~\cite{Dillig23}, as well as simply-typed programs~\cite{tata,ecta}.
However, these representations are insufficient for synthesis under refined library and query specifications, where correctness depends on enforcing logical implication and semantic relationships between subterms.~\footnote{A detailed elaboration of these points can be found in the Appendix.} \\  
\begin{wrapfigure}{r}{.20\textwidth}
\vspace*{-.25in}
\includegraphics[scale=.45]{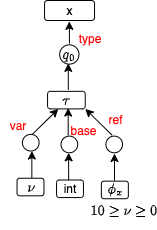}
\caption{A simple LTA for binding (\S{x} : \{ $\nu$ : \S{int} $\mid$ $10 \geq \nu \geq 0$ \})} 
\label{fig:baseqta}
\vspace*{-.25in}
\end{wrapfigure}
To address these limitations, we introduce a new data structure,
\automata{} (LTA), that allows us to capture such logical constraints.
An LTA supports a richer alphabet than other FTA variants by
incorporating logical qualifiers from decidable first-order theory
fragments. While it also allows constraints on its transitions similar
to other constrained tree automata~\cite{tata,ecta,cata}, it
additionally supports semantic constraints (e.g., logical entailment),
rather than being limited to syntactic equality or dis-equality
constraints~\cite{tata,ecta}.

For instance, Fig.~\ref{fig:compare-single} presents an LTA that
captures the space of terms represented by the following sentence
\[\{ f (t_1 , t_2)\ |\ t_1, t_2 \in \{ \phi_1, \phi_2, \phi_3 \} \wedge
\ t_1 \implies t_2 \}\] where both sub-trees are constrained using a
\textit{logical entailment} constraint (defined later) on the
transition (\textcolor{red}{{\sf l} $\vDash ${\sf r}}).  States are
depicted as circles and transitions as rectangles. Each transition may
have zero or more incoming states, and each incoming edge is annotated
with a position label, shown in red in the figure.  Here,
\textcolor{red}{l} and \textcolor{red}{r} are variables that capture a
specific position in the automaton and the constraint restricts which
choice of $t_1$ and $t_2$ are acceptable.  \NEW{The automaton accepts
  terms $f (\phi_2, \phi_1)$ and $f (\phi_3, \phi_1)$ since, in both
  cases, the constraint $\phi_i \implies \phi_j$ holds; however, it 
  rejects other syntactically valid terms like $f (\phi_1, \phi_2)$
  and $f (\phi_1, \phi_3)$ where the constraint does not hold.}

\paragraph{Representing Refinement Types using LTA.}
The above characteristics make LTAs effective at providing a compact
embedding for rich type structures. Fig.~\ref{fig:baseqta}
illustrates this, showing a transition in a trivial LTA representing a
variable \S{x} with refinement type \{ $\nu$ : \S{int} $\mid$ $10 \geq
\nu \geq 0$ \}.  The automaton rooted at state $q_0$ corresponds to
the refinement type \{ $\nu$ : \S{int} $\mid$ $10 \geq \nu \geq 0$ \},
while the overall LTA represents a binding \S{x} : \{ $\nu$ : \S{int}
$\mid$ $10 \geq \nu \geq 0$ \}.

\paragraph{Typing Semantics as LTA Transition Constraints.}
For component-based synthesis, an LTA must accept only
well-refinement-typed programs. LTA construction therefore mirrors
typing derivations, ensuring that ill-typed programs are excluded by
construction via LTA transition constraints.

Fig.~\ref{fig:semqta} shows a portion of the LTA for the example
library, embedding the semantics for function application. Transition
\S{app} corresponds to a function application expression in  a
standard refinement-typed calculus. Incoming state $q_f$ models the
set of all single-argument library functions and their types. We have
grayed out functions other than \S{decr} for clarity.  The
\textcolor{red}{type} edge for function \S{decr}, has two children,
one for its input \textit{argument} type and another for its
\textit{return} type, with location labels \textcolor{red}{in} and
\textcolor{red}{out} respectively.  State $q_a$ models the set of all
scalar arguments (showing only \S{x}).  The argument's type is
represented in the usual way for a base refinement type as shown in
Fig.~\ref{fig:baseqta}.  The application typing semantics requires
that the type of an actual argument be a subtype of the corresponding
formal argument type. This requirement is captured by constraints
\textcircled{a} and \textcircled{b}. Constraint \textcircled{b} enforces equality
of the underlying base types, while constraint \textcircled{a} ensures
logical implication between the corresponding refinement formulas.
\begin{wrapfigure}{r}{.50\textwidth}
\includegraphics[scale=.45]{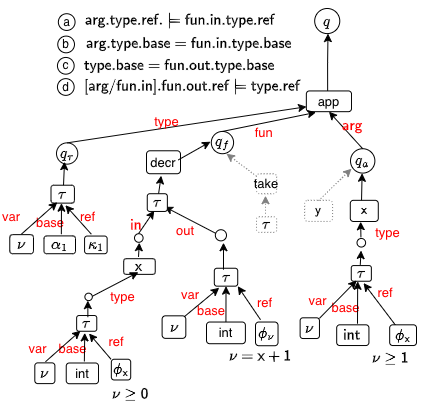}
\caption{An LTA fragment embedding a refinement typing semantics for function application.
} 
\label{fig:semqta}
\vspace*{-.4in}
\end{wrapfigure}
Additionally, a typing semantics often establishes logical entailment
constraints between formulas under a substitution. A good example is
the inferred type for an \S{app} term and the function's return
type. LTAs model this using the constraints \textcircled{c}, that equates
base types, and \textcircled{d}, that captures the logical entailment
(\S{fun.out.ref} $\vDash$ \S{type.ref}) under the usual mapping of
the actual argument to the formal for each possible choice of the
function and the argument, modeled using a compact substitution
relation over variables using a substitution over positions, shown as
\textcolor{red}{[} \S{arg}/\S{fun.in}\textcolor{red}{]}.

\paragraph{Efficient Pruning and Enumeration.}
We exploit LTA structure to develop a CBS algorithm that combines two
complementary reduction strategies.  First, we prune automaton states
that cannot participate in any well-typed program.  Second, we
identify semantic similarity between program fragments using
refinement subtyping, enabling eager pruning of logically similar
paths during search through merging of logically equivalent
subautomata.  Our key insight is that the intersection operation on
tree automata generalizes naturally to a notion of \textit{semantic}
intersection over LTAs, which enables efficient enumeration even when
feasible solutions are sparse.  This allows the synthesizer to prune
large regions of the search space that differ syntactically but are
equivalent or subsumed semantically.

We implement this approach in a tool called \name.  Our evaluation
shows that \name{} can synthesize programs for complex, highly
constrained CBS queries that are beyond the reach of existing systems,
while maintaining soundness and completeness.  This
  paper makes the following contributions:
\begin{itemize}

\item We address scalability and expressivity limitations in existing
  CBS frameworks, especially in the presence of queries with rich
  refinement-type specifications.
 
\item We introduce \emph{Liquid Tree Automata} (LTA), a tree automata
  variant that allows direct embedding of refinement-type
  specifications.

\item We develop a novel synthesis procedure over LTAs, incorporating
  pruning and similarity-based reductions, and prove the procedure
  sound and complete.

\item We implement our approach in \name{}, a CBS tool for OCaml
  programs; our evaluation results demonstrate significant scalability
  improvements over existing systems.

\end{itemize}

\section{Liquid Tree Automata}
\label{sec:qta}  

\begin{wrapfigure}{r}{0.3\textwidth}
  \vspace*{-.2in}
\small \begin{tabular}{l l l} 
& $\epsilon, p, p_i$ & $\in$ \emph{Position}  \\
$\psi_a$  & $ \in \textit{Atoms}$ ::=& \S{true} $\mid$ \S{false} \\

& $\mid$ $p$ = $p$  & (\textit{Syntactic}) \\
& $\mid$ $p$ $\vDash$ $p$ & (\textit{Semantic}) \\ 
$\theta$ & $\in$ \emph{Subs} := & [$\overline{p/p}$]  \\
$\psi$ & $\in \Psi$ =::&  \\
& $\psi_a$ $\mid$ $\neg$ $\psi$ & \\
& $\mid$ $\psi$ $\wedge$ $\psi$ $\mid$ $\psi \lor \psi$ & \\
& $\mid$ $\theta.\psi$ & \\
$\sigma$  & $\in Schema$ ::=& $\star$ = $\star$ $\mid$ $\star$ $\vDash$ $\star$\\ 
\end{tabular}%
\vspace*{-.1in}
\caption{LTA Constraints}
\label{fig:constraints}
\vspace{-.25in}
\end{wrapfigure}
The target language of our synthesizer ($\calculus$) is a standard
A-normalized~\cite{flanagan} call-by-value refinement-typed
$\lambda$-calculus~\cite{liquidextended} with constructors, constants
and variables, conditional expressions, and function abstraction and
application. To simplify the presentation, we assume all variables
have a single unique binding-site.  $\lambda_{\S{lta}}$ types include
standard base types like \S{int, bool} etc., along with algebraic
types like \S{lists} and \S{trees} over these base types and type
variables $\alpha_1, \alpha_2$, etc. Refinement types $\tau$, include
\emph{base refinements} and \emph{arrow refinements}. A base
refinement \{ $\nu : \S{t} \mid \phi$ \} qualifies a term of base
type \S{t} with a refinement qualifier $\phi \in \Phi$. An arrow
refinement refines a function type, where the argument \S{x} can occur
free in the qualifier. The set of qualifiers ($\Phi$) consist of
first-order predicate logic formulae over base-typed variables along
with method predicates ($Q$), which are user-defined, uninterpreted
function symbols such as {\sf len} and {\sf ord} over lists used in
our motivating example. A type context $\Gamma$
records term variables and library functions $g$ with their types. It
also records a set of propositions relevant to a specific context.

\begin{definition}[Positions in a term]
\label{def:pos-term}
A position $p$ in a term $t$ is of the form $i.j.k....n$, a sequence
of positive integers describing a path from the root of $t$ to a
sub-term. This describes what symbols are present at each position,
relative to the root.
\end{definition}
For easier comprehension, we give human-readable
labels to each number in a position. For example, consider
Fig.~\ref{fig:semqta} - the sequences used in the
constraints like \S{arg.type.base = fun.type.in.base} are
positions. \S{type.ref} is a synonym for position \S{1.3},
\S{fun.in.type.base} for \S{2.1.1.2}, etc.

\paragraph{\bf Constraints in \automata}
An LTA constraint $\psi$ is a predicate on terms in
\calculus. It is defined inductively over positions and Boolean
connectives, as shown in Fig.~\ref{fig:constraints}.
A valid \textit{atomic constraint} $\psi_a$ includes Boolean constants
like \S{true} and \S{false}, as well as \textit{syntactic equality}
between positions, given by $p$ = $p$ and
\textit{semantic entailment}  over positions, given by
$p \vDash p$. 
A \textit{constraint} $\psi$ is either a $\psi_a$ or a
constraint generated using Boolean connectives over other constraints.
{In several cases, these constraints require variable substitution at one position with another; this is captured using a list of substitutions \emph{Subs} ($\theta$), and a constraint $\theta.\psi$}.
Consequently, we can also classify \textit{atomic constraints} into
\textit{kinds} using two constraint schemas,
$\star$ = $\star$ and $\star$ $\vDash$ $\star$, respectively.

\begin{definition}[Liquid Tree Automata]
 A \textit{Liquid Tree Automata}, $\mathcal{A}$ defined over a finite ranked alphabet $\mathcal{F}$ derived from \calculus, is a tuple ($Q$, $\mathcal{F}$, $Q_f$, $\Delta$), where:
 \begin{itemize}
     \item $Q$ is a finite set of states.
     \item $Q_f \subseteq Q$ is a set of final states.
     \item $\Delta \subseteq Q^n \times \mathcal{F} \times \Psi \mapsto Q$, is a set of constrained \emph{transitions}. Each transition rule is of the form $f (q_1, q_2,...q_n) \T{\psi} q$,
  where $f \in \mathcal{F}$ with arity $n$, and a set of
  states $q_1, q_2, \ldots q_n \in Q$ and $\psi \in \Psi$ is a valid constraint. Here $q$ is the \textit{target state}.
 \end{itemize}
\end{definition}


\paragraph{\bf Language of a \automaton $\denotation{\mathcal{A}}$}
The language accepted by a LTA $\mathcal{A}$, is the set of all terms
in \calculus with some successful run of $\mathcal{A}$. This is, in fact, the set of all well-typed \calculus terms, constructed using the
methods found in a library. 
\begin{wrapfigure}{l}{.45\textwidth}
  \vspace*{-.2in}
  \small
 \centering 
\begin{tabular}{l  l l} 
$\denotationNode{q}$ &::=  $\bigcup_{i}$ ($\denotationEdge{\delta_i}$) \\ &  where $\delta_i$ = \\
& ($f ({q_i}_1, {q_i}_2$,\ldots,${q_i}_n)$ $\T{\psi}$ $q$) \\
&  \\
$\denotationEdge{\delta}$  &::=  \{ $f$ 

$\overline{t_i}$ $\mid$  $t_i$ $\in \denotationNode{q_i}$,$f$ 

$\overline{t_i} \vDash \psi$, \\ 
& $i \in [1\ldots n]$ \}, where \\
&$\delta$ = ($f ({q}_1, {q}_2$,\ldots,${q}_n)$ $\T{\psi} q$) \\
$\denotationEdge{\delta_{\bot}}$ &  := $\varnothing$ \\
& \\
$\denotation{\mathcal{A}}$  ::= & $\bigcup_{i}$ \{$\denotationNode{q_i}$ $\mid$ $q_i \in Q_f$\}\\
\end{tabular}%
\caption{LTA denotations.}
\label{fig:denotation}
\vspace*{-.3in}
\end{wrapfigure}
Formally, we define the language accepted
by $\mathcal{A}$ using its denotation $\denotation{\mathcal{A}}$ (see
Fig.~\ref{fig:denotation}). The denotation of a state $q$,
$\denotationNode{q}$ is the set union of the denotations of each of
the transitions $\delta_i$, $\denotationEdge{\delta_i}$, for which $q$
is a target state. The denotation of a transition $\delta$
builds a set of all terms, using the symbol at the current transition
($f$) and terms in the denotation of states incoming in $\delta$, 
filtering all terms that do not satisfy the transition
constraint $\psi$. Symbols $f \in \mathcal{F}$ include all \calculus terms including variables (\S{x}), constants (\S{c}),  conditional expressions ( \S{if} b \ then \ e \ else \ e), function abstractions and applications, and 
types ($\tau$). 

Intuitively, the satisfaction of a constraint by a
term $t \vDash \psi$ maps syntactic equality constraints to equality
of symbols and semantic entailment between qualifiers to logical
entailment of FOL formulas.\NEW{We also have a special bottom transition $\delta_{\bot}$, whose denotation is an empty set}. Since an LTA can have multiple final
states given by the set $Q_f$, the language it accepts is the union of
the denotation for all its final states.

\section{Synthesis using Liquid Tree Automata}
\label{sec:synthesis}
We formalize synthesis using LTAs by defining a consistency relation
between a typing environment  $\Gamma$ and a LTA $\mathcal{A}$.

\NEW{\begin{definition}[Consistency between LTAs and Type Environments]
A type environment $\Gamma$ is \textit{consistent} with an LTA $\mathcal{A} = $ ($Q$, $\mathcal{F}$, $Q_f$, $\Delta$) iff\\ $\forall e$, $\Gamma (e)$ = $\tau \iff 
\exists \mathcal{A}'$ such that $\mathcal{A}'$ is a sub-automaton~\footnote{A sub-automaton is defined using a subset definition over states and transitions.} of $\mathcal{A}$ and $e \in \denotation{\mathcal{A}'}$.
\end{definition}}

\begin{definition}[Synthesis Problem with LTAs] Given a type environment $\Gamma$ that relates library functions $f_i = \lambda (\overline{x_{i,j}}). e_{f_i}$ with their refinement types
$f_i : \overline{(x_{i,j} : \tau_{i,j})} \rightarrow \{ \nu :  t_i \mid \phi_{i} \} \in \Gamma$, and a synthesis query  $\varphi = \overline{(y_i : \tau_i)} \rightarrow \{\nu : t \mid \phi \}$, 
a solution to a CBS problem is a LTA $\mathcal{A}$, such that forall all $e \in \denotation{\mathcal{A}}$, $\Gamma \vdash e : \varphi$ and $\Gamma$ is consistent with $\mathcal{A}$.
\end{definition}

\subsection{Synthesis Algorithm}
The main synthesis algorithm, {\sc LTASynthesize} is shown in
Algorithm~\ref{alg:enumerate}.  It takes as input an alphabet
$\mathcal{F}$, which includes symbols from \calculus and an annotated
library of functions, a synthesis query specification $\varphi$, and a
bound on the size of the terms ${k}$ to synthesize.

The algorithm works in two phases: (1) an \textit{exploration} phase
adds states and transitions, expanding the automata. The resulting LTA
is then pruned/reduced by (2) a reduction phase.  The algorithm also
keeps track of \textit{similar} but not yet reduced transitions
through an equivalence set $\mathcal{E}$.
The output of the algorithm is a pair consisting of (i) an
LTA, $\mathcal{A}_{\S{min}}$ for the synthesis query based on
$\mathcal{F}$ and the typing semantics of \calculus,
and (ii) a set of solution terms in \calculus, possibly using $\mathcal{F}$ that satisfies the query specification. 
The algorithm returns a failure value ($\bot$) if it cannot find a solution within the given max-depth $k$.

The algorithm begins (line~1) by initializing the similarity set
$\mathcal{E}$ to $\varnothing$ and constructing an initial LTA
$\mathcal{A}_0$ using the well-formedness procedure {\sc WF}.
Given the library $\mathcal{F}$ and an empty automaton
$\mathcal{A}_{\bot}$, {\sc WF} deterministically constructs an LTA
that represents all well-formed terms of size one consistent with the
typing rules in Fig.~\ref{fig:typing}, including the query
specification $\varphi$.  In particular, initialization introduces
states and transitions corresponding to query arguments, library
functions, base refinement types, and a distinguished final transition
encoding the top-level constraint of $\varphi$.

At line~2, the algorithm checks whether the language of
$\mathcal{A}_0$ is non-empty using procedure {\sc NEmpty}
(lines~14-15), which identifies final states whose denotation is
non-empty.  If such states exist, the algorithm extracts the set of
solution terms using the LTA denotation $\denotation{\cdot}$
(Fig.~\ref{fig:denotation}) and returns
$(\mathcal{A}_0, \denotation{\mathcal{A}_0})$ (line~3).
Otherwise, the algorithm invokes procedure {\sc Enumerate}
(lines~5-13), which performs iterative exploration and reduction of
the search space.  If {\sc Enumerate} is invoked at depth $k$, it
returns $\bot$, indicating synthesis failure.\\
When the depth of $\mathcal{A}$ is less than the bound $k$,
{\sc Enumerate} first enters an exploration phase by extending
$\mathcal{A}$ with new transitions using procedure {\sc Transition}
(line~6).  This step applies the typing rules in
Fig.~\ref{fig:typing} to construct larger well-typed terms.
\begin{wrapfigure}{r}{.5\textwidth}
  \vspace*{-.1in}
  \centering
  \hspace*{.17in}
  \begin{algorithm*}[H]
    \small
\SetAlgoNoLine
\SetNlSty{texttt}{(}{)}
\SetKwFunction{LTAsynthesize}{\textsc{LTASynthesize}}
\Indm\LTAsynthesize{$\langle \mathcal{F}, \varphi = \overline{(\mathsf{x_i} : \tau_i)} \rightarrow \{ \S{v} : \S{t} | \phi\}, k \rangle$}\\
\Indp
     \tcp{\textcolor{cyan}{Initialize}}
    \nl$\mathcal{A}_{0}$ $\leftarrow$ \textsc{WF} ($\mathcal{F}$, $\mathcal{A}_{\bot}$); $\mathcal{E} \leftarrow \varnothing$ \\
    \tcp{\textcolor{cyan}{Check solution in Initial $\mathcal{A}_0$}}
    \nl\If { $\hat{Q} = $ {\sc NEmpty} ($ \mathcal{A}_0$)}{  
             \nl {\bf return} ($\mathcal{A}_0$, $\bigcup_{q \in \hat{Q}} \denotation{\mathcal{A}_q}$)
    } 
    \tcp{\textcolor{cyan}{Iteratively explore-reduce-check}}       
    \nl {\bf return} \textsc{Enumerate} ($\mathcal{A}_0$, $\varphi$, $k$)\\[1mm]
\SetKwFunction{enumerate}{\textsc{Enumerate}}

\Indm\enumerate{$\langle \mathcal{F}, \mathcal{A}, \varphi = \overline{(\mathsf{x_i} : \tau_i)} \rightarrow \{ \S{v} : \S{t} | \phi\}, k \rangle$}\\
\Indp
        \nl \eIf{depth ($\mathcal{A}$) < k}    
        {  
           \nl $\mathcal{A}$ $\leftarrow$ \textsc{Transition} ($\mathcal{F}$, $\mathcal{A}$);\\ 
           \nl $\mathcal{A}$ $\leftarrow$ \textsc{Prune} ($\mathcal{A}$);\\ 
           \nl $\mathcal{E}$ $\leftarrow$ \textsc{Similarity} ($\mathcal{A}$, $\mathcal{E}$); \\ 
           \nl ($\mathcal{A}_{\S{min}}$, $\mathcal{E}$) $\leftarrow$ \textsc{Minimize} ($\mathcal{A}$, $\mathcal{E}$); \\ 
           \nl\If { $\hat{Q} = $ {\sc NEmpty} ($\mathcal{A}_{\S{min}}$)}{  
             \nl {\bf return} ($\mathcal{A}_{\S{min}}$, $\bigcup_{q \in \hat{Q}} \denotation{\mathcal{A}_q}$)
             }   
           \nl \textsc{Enumerate} ($\mathcal{A}_{\S{min}}$, $\varphi$, $k$)\\
        }{ 
                \nl {\bf return} $\bot$ \\[1mm]
         }
\SetKwFunction{nonemptyqta}{\textsc{NEmpty}}

\Indm\nonemptyqta{$\langle \mathcal{A} \rangle$}\\
\Indp
    \nl $\hat{Q}$ $ \leftarrow \{ q_f  \mid q_f \in Q_f, \denotationNode{q_f} \neq \varnothing$ \} \\
    \nl {\bf return} $\hat{Q}$
    
    \caption{Synthesis Algorithm.}
\label{alg:enumerate}
  \end{algorithm*}
  \vspace*{-.4in}
\end{wrapfigure}

The algorithm then enters a reduction phase.  It first applies
{\sc Prune} (line~7) to eliminate unreachable or ill-typed portions of
the automaton.  Next, similarity relations between transitions are
identified using {\sc Similarity} (line~8), producing an updated
equivalence set $\mathcal{E}$.  Finally, {\sc Minimize} (line~9)
applies the minimization rules ({\sc M-Trans}) and ({\sc M-LTA})
from Fig.~\ref{fig:similarity} to merge semantically redundant
states and transitions, yielding a reduced automaton
$\mathcal{A}_{\S{min}}$.

At line~10, the algorithm again checks whether the language of
$\mathcal{A}_{\S{min}}$ is non-empty.  If so, it returns
$(\mathcal{A}_{\S{min}}, \denotation{\mathcal{A}_{\S{min}}})$
(line~11).  Otherwise, the algorithm recurses on
$\mathcal{A}_{\S{min}}$, continuing the explore--reduce--check loop
until a solution is found or the depth bound is exceeded.

\subsection{LTA Construction}
The transitions $\Delta$ for a LTA $\mathcal{A}$ are constructed
using the \textbf{Well-formedness} and \textbf{Transitions} judgments given in Fig.~\ref{fig:typing}.
The latter judgment holds if, given library $\mathcal{F}$ and  automata $\mathcal{A}$,  a new n-ary transition  can be added to $\mathcal{A}$  corresponding to an n-ary symbol $f \in \mathcal{F}$, with $q_1, q_2, ...,q_n$ being the incoming states in the transition and $q$ being the target state, such that $\psi$ captures the typing constraints for the valid terms in the language of the transition $\denotationEdge{.}$ 
\footnote{For perspicuity,  we have elided several details about polymorphic transition construction and type inference; these mostly concern proper scope management within an LTA. 
The complete set of rules can be found in the Appendix.}
We begin by generalizing the definition of \emph{term at a position} (Definition~\ref{def:pos-term}) to LTAs. Given a state $q$ (or transition $\delta$), and a position $p$, we define \textit{LTA states at a position} (equivalently, LTA transitions at a position) $p$ from $q$ (or $\delta$), denoted as $q \blacktriangleright p$ (or $\delta \blacktriangleright p$):

\begin{definition}[Transitions at a position, $q \blacktriangleright p$]
For each $\delta_i$ = ($f ({q_i}_1, {q_i}_2$,\ldots,${q_i}_n)$ $\T{\psi}$ $q$).
$q \blacktriangleright \epsilon$ = \{$q$\} and $q \blacktriangleright j.p$ = $\bigcup_i (\delta_i \blacktriangleright j.p$), where $\delta \blacktriangleright j.p$ = $q_j \blacktriangleright p$ if $j \leq n$ (arity of $\delta_i$) and $\varnothing$ otherwise.
\end{definition}

\begin{figure*}[htbp]
\begin{flushleft}
\bigskip
{\bf Well-formedness}\quad \fbox{\footnotesize
    $\begin{array}{c} 
    \mathcal{F}, \mathcal{A} \vdash^{\S{wf}} f \in \mathcal{F} (\overline{q_i}) \T{\psi} q      
    \end{array}$
}
\end{flushleft}
\bigskip
\begin{minipage}{0.20\textwidth}
{\footnotesize
\begin{center}
\inference[{\sc wf-prim}]{\S{t} \in \S{T}_{\mathcal{F}} & q_{\S{t}} \notin Q}{\mathcal{F}, \mathcal{A} \vdash^{\S{wf}} \S{t} () \T{} q_{\S{t}}}
\end{center}
}
\end{minipage}
\begin{minipage}{0.20\textwidth}
{\footnotesize
\begin{center}
\inference[{\sc wf-pred}]{\phi \in \Phi_{\mathcal{F}} & q_{\phi} \notin Q}{\mathcal{F}, \mathcal{A} \vdash^{\S{wf}} \phi \T{} q_{\phi}}
\end{center}
}
\end{minipage}
\begin{minipage}{0.20\textwidth}
{\footnotesize
\begin{center}
\inference[{\sc wf-var}]{x \in \S{Vars}_{\mathcal{F}} & q_{x} \notin Q}{\mathcal{F}, \mathcal{A} \vdash^{\S{wf}} x \T{} q_{x}}
\end{center}
}
\end{minipage}
\begin{minipage}{0.4\textwidth}
{\footnotesize
\begin{center}
\inference[{\sc wf-base}]{(\tau \equiv \{ x : t \mid \phi \}) \in \tau_{\mathcal{F}} & q_{\tau} \notin Q} 
{\mathcal{F}, \mathcal{A} \vdash^{\S{wf}} \tau (q_x, q_t, q_{\phi}) \T{} q_{\tau}}
\end{center}
}
\end{minipage}
\bigskip
\begin{minipage}{0.4\textwidth}
{\footnotesize
\begin{center}
\inference[{\sc wf-arrow}]{(\tau_{\rightarrow} \equiv \tau_i \rightarrow \tau_j) \in \tau_{\mathcal{F}} \\
\tau_i, \tau_j \in \tau_{\mathcal{F}}  & q_{\tau_{\rightarrow}} \notin Q} 
{\mathcal{F}, \mathcal{A} \vdash^{\S{wf}} \tau_{\tau{\rightarrow}} (q_{\tau_i}, q_{\tau_j}) \T{} q_{\tau_{\rightarrow}}}
\end{center}
}
\end{minipage}
\begin{minipage}{0.4\textwidth}
{\footnotesize
\begin{center}
\inference[{\sc Q-goal}]{\varphi = \overline{(\mathsf{x_i} : \tau_i)} \rightarrow \tau & q_{\S{term}_k}, q_{\tau} \in Q \\ q_{\S{goal}} \in Q_f \\
\psi = \textsc{SubType} (q_{\S{term}_k}\blacktriangleright{\S{type}}, q_{\S{goal}}\blacktriangleright{\S{type}})
}
{\mathcal{F}, \mathcal{A}  \vdash^{\S{wf}} \S{goal} (q_{\tau}, q_{\S{term}_k}) \T{\psi} q_{\S{goal}}}
\end{center}
}
\end{minipage}

\begin{flushleft}
{\bf Transitions}\quad \fbox{\footnotesize
    $\begin{array}{c} \mathcal{F}, \mathcal{A} \vdash  f \in \mathcal{F} (\overline{q_i}) \T{\psi} q \\ 
    \end{array}$
}
\end{flushleft}
\bigskip

\begin{minipage}{0.6\textwidth}
{\footnotesize
\begin{center}
\inference[{\sc e-app}]{q_f, q_a \in Q &  \tau (\overline{q_i}) \T{} q_{\tau} \in \Delta\\
\theta = [q_{\S{a}}\blacktriangleright{\epsilon}/q_{f}\blacktriangleright{\S{in}}] \\
\psi = \textsc{SubType} (q_{\S{a}}\blacktriangleright{\S{type}}, q_{f}\blacktriangleright{\S{in}}) 
\wedge \\
\ \ \ \ 
\theta. \textcolor{black}{(}\textsc{SubType} (q_{f}\blacktriangleright{\S{out}}, q_{\S{app}}\blacktriangleright{\S{type}}) \textcolor{black}{)}}
{ \mathcal{F}, \mathcal{A} \vdash \S{app} \ (q_{\tau}, q_f, q_a) \T{\psi} q_{\S{app}}}
\end{center}
}
\end{minipage}
\begin{minipage}{0.25\textwidth}
{\footnotesize
\begin{center}
\inference[{\sc e-var}]{ x : \tau \in \mathcal{F} & q_x \notin Q}
{\mathcal{F}, \mathcal{A} \vdash x (q_{\tau}) \T{} q_x}
\end{center}
}
\end{minipage}

\bigskip

\begin{minipage}{0.4\textwidth}
{\footnotesize
\begin{center}
\inference[{\sc e-const}]{ \vdash c : \tau \in \mathcal{F} \\ q_c \notin Q}
{\mathcal{F}, \mathcal{A} \vdash c (q_{\tau}) \T{} q_c}
\end{center}
}
\end{minipage}
\begin{minipage}{0.4\textwidth}
{\footnotesize
\begin{center}
\inference[{\sc e-if}]{q_b, q_t, q_f \in Q & \tau (\overline{q_i}) \T{} q_{\tau} \in \Delta \\
\psi =  ((q_b \blacktriangleright {\S{ref}}) \wedge \textsc{SubType} (q_{t}\blacktriangleright {\S{type}}, q_{\S{if}}\blacktriangleright{\S{type}})) \wedge \\
\   \   \   \  (\neg (q_b \blacktriangleright{\S{ref}}) \wedge \textsc{SubType} (q_{f}\blacktriangleright{\S{type}}, q_{\S{if}}\blacktriangleright{\S{type}}))}{ \mathcal{F}, \mathcal{A} \vdash \S{if} \ (q_{\tau}, q_b, q_t, q_f) \T{\psi} q_{\S{if}}}
\end{center}
}
\end{minipage}

\bigskip

{\footnotesize
\begin{equation}
\label{eq:sub}
\textsc{SubType} \textcolor{red}{(} \delta_i, \delta_j \textcolor{red}{)} = 
        \begin{cases}
            \textcolor{red}{(} \tau_i (\overline{q_i}) \T{\psi_i} q_{\tau_i}, 
            \tau_j (\overline{q_j}) \T{\psi_j} q_{\tau_j}\textcolor{red}{)} & \S{i.type.t} = \S{j.type.t}  \\& \wedge \ \S{i.type.ref} \vDash \S{j.type.ref} \\
            \textcolor{red}{(}\tau_{\rightarrow_i} (\overline{q_i}) \T{\psi_i} q_{\tau_{\rightarrow_i}}, 
            \tau_{\rightarrow_j} (\overline{q_j}) \T{\psi_j} q_{\tau_{\rightarrow_j}}\textcolor{red}{)} &  \textsc{SubType} \textcolor{red}{(} \delta_j\blacktriangleright{\S{in}}, \delta_i\blacktriangleright{\S{in}}\textcolor{red}{)}   \\ & \wedge \  \textsc{SubType} \textcolor{red}{(} \delta_i\blacktriangleright{\S{out}}, \delta_j\blacktriangleright{\S{out}} \textcolor{red}{)} \\
            \textcolor{red}{(} \_, \_ \textcolor{red}{)} & \S{true}
        \end{cases}
\end{equation}
}
\caption{Selected rules for constructing transitions $\Delta$,  basis for {\sc WF} and {\sc Transition}.}
\label{fig:typing}
\end{figure*}

\subsubsection{Well-formedness Rules}



Given an alphabet $\mathcal{F}$ (library and query), a current LTA $\mathcal{A}$, and the well-formedness typing semantics of \calculus, we construct LTAs using judgments of the form
$\mathcal{F}, \mathcal{A} \vdash^{\S{wf}} f \in \mathcal{F}(\overline{q_i}) \T{\psi} q$,
which determine when a transition may be added to $\mathcal{A}$.
These rules characterize when leaf and composite transitions may be added for base types, predicates, variables, refinement types, arrow types, and polymorphic types, and directly mirror the well-formedness typing rules of \calculus.

To illustrate the behavior of these rules, Fig.~\ref{fig:wf-example} shows the application of the well-formedness and transition rules from Fig.~\ref{fig:typing} over a small example library of terms given below:
\begin{tabbing}
\noindent $\mathcal{F}$ =  \{\=\,\S{ f : (n : int)} $\rightarrow$ \S{(l : \{ v : [a]}\ |\  \S{len (v) > 0} \}) $\rightarrow$ \{ \S{v : [a]}\ |\ \S{len (v) > 0 }\}, \\
\>\S{xs : \{ v : [int]\ |\ len (v) > 0\} ; ys : [bool] \} }\}
\end{tabbing}
Assume that our query $\varphi$ is given by the refinement type:
\begin{tabbing}
\hspace*{.2in}(\S{xs} :  \{ v : [int] | len (v) > 0\}) 
$\rightarrow$ \{$\nu$ : [int] | len ($\nu$) > len (xs)\}
\end{tabbing}
\noindent and the value of $k$, the term size is 2.

Rule {\sc wf-prim} (the corresponding LTA construction is shown using the box labeled with \textcircled{3} in Fig.~\ref{fig:wf-example}) introduces a state $q_{\S{t}}$ and a nullary transition labeled $\S{t}$ for each primitive type $\S{t}$, like \S{char}, \S{int}, etc., in the language.
Rule {\sc wf-pred} similarly adds leaf transitions for
refinement formula $\phi$ in some annotation in $\mathcal{F}$ or in the
query. 
Rule {\sc wf-base} (shown with \textcircled{2}) introduces a state $q_{\tau}$ for each well-formed base refinement type $\tau = \{ x : t \mid \phi \}$ and adds a ternary transition with incoming states $q_x$, $q_t$, and $q_{\phi}$ corresponding to the bound variable, base type, and refinement formula, respectively. These states are provided by earlier applications of {\sc wf-var}, {\sc wf-prim}, and {\sc wf-pred}.

Rule {\sc wf-arrow}(\textcircled{1}) generates similar
transitions for each arrow refinement type, with a symbol
$\tau_{\rightarrow}$ and two incoming states for the argument-type and
the result-type for the arrow.
Rule {\sc Q-goal} introduces a distinguished final state $q_{\S{goal}}$ and a \S{goal} transition connecting sub-automata rooted at $q_{\S{term}_k}$, representing candidate terms of size $k$. The associated constraint $\psi$ enforces that the synthesized term’s return type is a subtype of the query’s annotated return type, as generated by {\sc SubType}.

Our extended \calculus and typing semantics also include \textit{type} and \textit{refinement} abstractions, allowing parametric polymorphism in the refinement types setting~\cite{relref,liquidextended}. Consequently, the above rules also extend naturally to these abstractions and their corresponding type and refinement applications.



\subsubsection{Expression Transition Rules}
Transition judgments are similar in structure to the well-formedness judgments, but simulate the refinement type judgments for the expressions and types in the \calculus. These rules specify how to add transitions
corresponding to \calculus expressions and types. Each
n-ary expression transition has n+1 incoming
states, with the state at the position zero capturing the sub-automata
for the possible types of the expression.\\
\begin{wrapfigure}{r}{.50\textwidth}
\vspace*{-.4in}
\includegraphics[scale=.60]{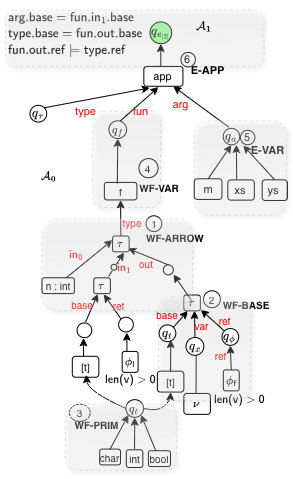}
\caption{Application of {\sc Well-formedness} and {\sc Transition} rules to construct initial LTA $\mathcal{A}_0$  and $\mathcal{A}_1$ for the example library.} 
\label{fig:wf-example}
\vspace*{-.4in}
\end{wrapfigure}
Rules {\sc E-var}(construction shown in Fig.~\ref{fig:wf-example}, \textcircled{5}) and {\sc E-const} introduce transitions for variables and constants together with their types. Rule {\sc E-var} (\textcircled{1}) applies to both scalar variables and variables bound to functions; accordingly, the associated type $\tau$ is represented either as a base refinement (via {\sc wf-base}) or as an arrow refinement (via {\sc wf-arrow}).




 

Rule {\sc E-app}(Figure~\ref{fig:wf-example}, \textcircled{6}) introduces transitions corresponding to function application expressions.
It assumes the states $q_f$ and $q_a$ are already present for
the function $f$ and argument $a$. 
It first constructs a transition for the application’s result type $\tau$ using the well-formedness judgments, with $q_{\tau}$ as the target state.

The inferred transition uses {\sf app} as its symbol and has three incoming states, corresponding to the result type ($q_{\tau}$), the function ($q_f$), and the argument ($q_a$).
The key aspect of this rule is the set of constraints $\psi$ associated with the transition. These constraints relate the three subautomata rooted at $q_{\tau}$, $q_f$, and $q_a$, and encode the standard refinement-typing semantics of function application~\cite{liquidextended}.
These constraints first apply the standard substitution of the possible formal arguments in the function with the actual arguments. This is done via creating a substitution $\theta$ between positions given formally as  [$q_{\S{a}}\blacktriangleright \epsilon /
q_{f}\blacktriangleright{\S{in}}$].
Second, the rule uses an auxiliary function {\sc SubType} (Equation~\ref{eq:sub}), which given two transitions, returns constraints sufficient to capture the subtype relation between their types. 
For example, $\psi$ in
{\sc E-app} captures two main relations. The first is a subtyping constraint
between a function's formal argument and the actual expression passed as
the argument ({\sc SubType} ($q_{\S{a}}\blacktriangleright{\S{type}},
q_{f}\blacktriangleright{\S{in}}$)). The second subtype relation is between
the function's result type and the type of the \S{app} expression, under the mapping $\theta$ between the formal and the actual arguments. This is given by
$\theta$. ({\sc SubType}($q_{f}\blacktriangleright{\S{out}},q_{\S{app}}\blacktriangleright{\S{type}}$)).

Note that
the {\sc E-App} rule is sufficiently general to allow synthesis to
support \textit{partial} (i.e., higher-order) function applications.
Rule {\sc E-if} introduces transitions for conditional expressions using subautomata for the Boolean condition (rooted at $q_b$) and the two branches (rooted at $q_t$ and $q_f$).
The associated constraints encode the standard refinement typing semantics of conditionals: the true branch is constrained under the assumption that the condition holds, while the false branch is constrained under its negation.


\section{LTA Reductions}
\label{sec:reductions}

\subsection{\textsc{Prune}}

The LTA formulation in Section~\ref{sec:qta} accepts only
well-typed terms from the \calculus. However, we can make the
synthesis procedure more efficient by eagerly reducing portions of the
automata (i.e. sub-automata) which are irrelevant to the construction
of any solution.  
\begin{figure*}[h!p]
\begin{flushleft}

{\bf Pruning}\quad\fbox{\footnotesize
     $\begin{array}{c} 
        \mathcal{A} \vdash \Delta \rightsquigarrow \Delta' \ \ \ \  \mid \ \ \ \ \mathcal{A} \vdash  (\delta_i, \delta_j) \rightsquigarrow^{\psi_a} \delta' 
      \end{array}$
           
} 

\end{flushleft}
\bigskip
{\footnotesize
\begin{minipage}{0.30\textwidth}
\begin{center}
\inference[{\sc p-trans}]{ \delta \equiv f (\ldots) \T{\psi_j \wedge \ldots} q & 
\psi_j \equiv p_1 (=\mid \vDash) p_2 \\
\delta_i \in \delta\blacktriangleright p_1 & \delta_j \in \delta\blacktriangleright p_2 \\
\Delta_{r} = \{ \delta_{i,r} \mid \mathcal{A} \vdash  (\delta_i, \delta_j)  \rightsquigarrow^{\psi_j} \delta_{i,r} \} 
} 
{\mathcal{A} \vdash \Delta \rightsquigarrow \Delta[\delta\blacktriangleright p_1 / \Delta_r]} 
\end{center}
\end{minipage}
%
\begin{minipage}{0.40\textwidth}
\begin{center}
\inference{\delta \equiv f (\ldots) \T{\psi_j \wedge \ldots} q & 
\psi_j \equiv p_1 = p_2 \\
\delta_{1} \in \delta\blacktriangleright p_1 & \delta_{2} \in \delta\blacktriangleright p_2 \\
\delta_{\S{r}} = {\sqcap}_{\S{Syntax}} (\delta_{1}, \delta_{2})
} 
 { \mathcal{A} \vdash  (\delta_{1}, \delta_2) \rightsquigarrow^{\psi_j} \delta_{\S{r}}}[{\footnotesize{\sc p-syn-eq}}]   
\end{center} 
\end{minipage}
\bigskip
\begin{minipage}{1.0\textwidth}
\begin{center}
\inference{\delta \equiv f (\ldots) \T{\psi_j \wedge \ldots} q & 
\psi_j \equiv \theta. p_1 \vDash p_2 \\
\delta_{1} \in \delta\blacktriangleright p_1 & \delta_{2} \in \delta\blacktriangleright p_2 &
\delta_{\S{r}} = {{\sqcap}^{\psi_j}}_{\S{Semantics}}(\delta_{1}, \delta_{2})
} 
 { \mathcal{A} \vdash   (\delta_{1}, \delta_2)  \rightsquigarrow^{\psi_j} \delta_{\S{r}}}[{\footnotesize{\sc p-sem-ent}}]  
\end{center} 
\end{minipage}
}
\bigskip

\begin{flushleft}
{\bf Similarity}\quad \fbox{\footnotesize
  $\begin{array}{c} 
\mathcal{A} \vdash^{\S{\bf sim}} \delta_i \lesssim \delta_j \ \ \ \  \mid \ \ \ \ \mathcal{A} \vdash \mathcal{E} \rightsquigarrow \mathcal{E}'
\end{array}$
}
\end{flushleft}
\bigskip
\begin{minipage}{0.4\textwidth}
{\footnotesize
 \inference[\footnotesize {\sc s-trans}]{ \psi_{<:} = \textsc{SubType} \textcolor{red}{(} \delta_i\blacktriangleright \S{type}, \delta_j\blacktriangleright \S{type} \textcolor{red}{)} \\  \delta_{\bot} \notin {{\sqcap}^{\psi_{<:}}}_{\S{Semantics}} (\delta_i\blacktriangleright \S{type} , \delta_j\blacktriangleright \S{type}) }
			  { \mathcal{A} \vdash^{\S{{\bf sim}}} \delta_i \lesssim \delta_j} 
}
\end{minipage}
\begin{minipage}{0.40\textwidth}
{\footnotesize
\inference{ (\delta_i, \delta_j) \notin \mathcal{E} \\
\mathcal{A} \vdash^{\mathbf{sim}} \delta_i \lesssim \delta_j } {\mathcal{A} \vdash \mathcal{E} \rightsquigarrow \mathcal{E} \cup \{(\delta_i, \delta_j)\}}[{\sc s-eq}]
}
\end{minipage}

\bigskip

\begin{flushleft}

{\bf Minimization}\quad\fbox{\footnotesize
     $\begin{array}{c} 
        (\mathcal{A}, \mathcal{E})  \vdash \Delta \rightsquigarrow \Delta' \ \ \ \  \mid \ \ \ \ \vdash (\mathcal{A}, \mathcal{E})   \rightsquigarrow (\mathcal{A}', \mathcal{E}') 
      \end{array}$
           
} 

\end{flushleft}
\bigskip
{\footnotesize
\begin{minipage}{0.40\textwidth}
\begin{center}
\inference[{\sc m-trans}]{\delta_i, \delta_j \in \Delta & (\delta_i , \delta_j) \in \mathcal{E} \\
\delta_i \equiv f (q_1, q_2, \ldots q_j \ldots q_n) \T{\psi_i} q_i  \\
\delta_j \equiv f' (q_1', q_2', \ldots q_m') \T{\psi_j} q_j  \\
\Delta' = \bigcup_k \textcolor{red}{\{} \delta_k[q_j / q_i] \\ 
 \   \    \ \mid \delta_k = \hat{f} (\hat{q_1}, {\bf q_j}, \ldots \hat{q_m} ) \hookrightarrow \hat{q_k}\textcolor{red}{\}} } 
{(\mathcal{A}, \mathcal{E}) \vdash \Delta \rightsquigarrow (\Delta \cup \Delta') \setminus 
\{\delta_j\}} 
\end{center}
\end{minipage}
\begin{minipage}{0.40\textwidth}
\begin{center}
\inference{ {\footnotesize \mathcal{A} \equiv (Q, \mathcal{F}, Q_f, \Delta}) \\ {\footnotesize  (\mathcal{A}, \mathcal{E}) \vdash \Delta \rightsquigarrow^{\star} \Delta'}} 
 { \vdash (\mathcal{A}, \mathcal{E}) \rightsquigarrow ((Q, \mathcal{F}, Q_f, \Delta'), \varnothing)}[{\footnotesize{\sc m-lta}}]   
\end{center} 
\end{minipage}
}
\bigskip

{\footnotesize
\begin{equation}
\label{eq:sem}
{{\sqcap}^{\psi}}_{\S{Semantics}}(\delta_i, \delta_j) = 
        \begin{cases}
             \delta_i & \sigma (\psi_j) = \theta. \star \vDash \star, \S{Symbol} (\delta_i) = \phi_i, \S{Symbol} (\delta_j) = \phi_j,  \\
                       & \S{Given \ the \ automaton} \ \mathcal{A} , \ s.t. \ \mathcal{R} (\Gamma, \mathcal{A}) \\
                       & \denotation{\Gamma} \wedge \denotation{\theta}. \denotation{\phi_i} \implies \denotation{\phi_j} \\
            \delta_{\bot} & \S{otherwise}
        \end{cases}
\end{equation}
}

\caption{Selective rules for similarity inference and LTA Minimization.}
\label{fig:similarity}
\end{figure*}



The inference rules for pruning irrelevant code, which underlie the
{\sc Prune} routine in Algorithm~\ref{alg:enumerate}, are shown in
Fig.~\ref{fig:similarity}. These rules are expressed using two
judgment forms.
The judgment $\mathcal{A} \vdash \Delta \rightsquigarrow \Delta'$
defines how a set of transitions $\Delta$ is reduced to a pruned set
$\Delta'$ within an LTA $\mathcal{A}$.
This reduction relies on a second judgment form,
$\mathcal{A} \vdash (\delta_i, \delta_j) \rightsquigarrow^{\psi_a} \delta'$,
which specifies how a pair of transitions $\delta_i$, $\delta_j$ are reduced under
an atomic constraint $\psi_a$ using syntactic or semantic intersection operation, yielding a possibly simplified transition
$\delta'$.

The {\sc p-trans} rule applies these atomic reductions to transition
sets.
Given a transition $\delta$ with constraint $\psi$, assumed to be a
conjunction of atomic constraints $\psi_j$ relating transitions at
positions $p_1$ and $p_2$, the rule updates the set of transitions at
position $p_1$ (denoted $\delta \blacktriangleright p_1$).
Each such transition is replaced by its reduced form, producing an
updated transition set $\Delta_r$ computed using the atomic reduction
rules described below.

Rule
{\sc p-syn-eq} handles syntactic equality constraints over positions
$p_1$ and $p_2$. It performs a syntactic intersection~\cite{ecta,tata} over the two
set of transitions at these position.  Syntactic intersection
($\sqcap_{\S{Syntax}}$) is a standard tree intersection
operation. Intuitively, it compares the two transitions for syntactic
equality of transition symbols while recursively intersecting each incoming state in the transition.
The rule updates each transition in $\delta\blacktriangleright p_1$, with its syntactic intersection with some transition in $\delta \blacktriangleright p_2$.

Rule {\sc p-sym-ent} handles semantic entailment constraints, where
$\psi_j$ has the form $\theta .\, p_i \vDash p_j$.
In this case, syntactic intersection is insufficient, since the
formulas associated with positions $p_1$ and $p_2$ cannot be compared
by syntactic equality alone.
We therefore introduce a semantic intersection operation,
denoted ${\sqcap^{\psi}}_{\S{Semantics}}$ shown in Equation~\ref{eq:sem},
which applies to transitions carrying refinement qualifiers.

Given two such transitions, the operation checks whether the refinement
formula at position $p_1$ (extracted using an auxiliary function \S{Symbol} over the transition) logically entails the formula at position
$p_2$ under the substitution $\theta$ and the interpretation of the environment $\Gamma$, shown as ($\denotation{\Gamma} \wedge \denotation{\theta}. \denotation{\phi_i} \implies \denotation{\phi_j}$).
If the entailment holds, the intersection yields the transition at the
more specific position $p_1$, otherwise, it returns a 
bottom transition $\delta_{\bot}$.
All bottom transitions are subsequently eliminated by normalization.
Intuitively, this operation compares transitions modeling
refinement formulas in two types, and keeps the sub-type transition~\footnote{
  A detailed example can be found in the accompanying Appendix.}.
\NNEW{The intersection operations form the basis of {\sc
  Prune} procedure in {\sc LTASynthesize}.}

\subsection{\textsc{Similarity} and \textsc{Minimize}}
\label{subsec:similarity}

The similarity inference rules are defined by the {\sc Similarity}
judgments in Fig.~\ref{fig:similarity}. The {\sc S-trans} rule
characterizes when two transitions are considered \emph{similar},
based on the relationship between their associated type sub-automata.

Similarity is determined by the constraint
$\psi_{<:} = \textsc{SubType}(\delta_i\blacktriangleright{\S{type}},
\delta_j\blacktriangleright{\S{type}})$, which checks whether the types
associated with the two transitions are related by the standard
subtyping relation induced by the typing semantics of \calculus.
Rule {\sc S-eq} enumerates pairs of transitions from $\mathcal{A}$ that
are not yet recorded in the similarity set, and adds to
$\mathcal{E}$ whenever the similarity constraint holds.

\emph{Similarity Reduction}. The similarity relation gives us a principled minimization strategy for
LTAs, in which transitions that are similar up to subtyping are merged,
retaining only the most specific representative.
The minimization rules in Fig.~\ref{fig:similarity} formalize this
process.

The first judgment form,
$(\mathcal{A}, \mathcal{E}) \vdash \Delta \rightsquigarrow \Delta'$,
defines how the transition set of an LTA is reduced using the similarity
relation $\mathcal{E}$.
As specified by Rule {\sc M-trans}, when two similar transitions are
identified, the transition corresponding to the supertype is removed,
while the transition corresponding to the subtype is retained.
All incoming edges that previously targeted the removed transition are
redirected to the retained one, preserving reachability within the
automaton.
The second judgment, formalized by Rule {\sc M-lta}, lifts transition
minimization to the level of entire LTAs.
It computes a minimized transition set $\Delta'$ by applying 
a transitive closure of the transition update rules.
The alphabet $\mathcal{F}$ of the automaton remains unchanged during
this process.
After minimization, the similarity set $\mathcal{E}$ is reset, since
all similar transitions have either been merged or eliminated.

Note that equivalence of terms (and of their corresponding sub-automata)
is strictly stronger than similarity.
As a result, similarity-based reduction also subsumes equivalence-based
pruning, eliminating both semantically equivalent and strictly
less-specific sub-automata.
This unified treatment enables aggressive reduction of redundant terms,
leading to efficient enumeration over a substantially reduced search
space.

\section{Handling Cycles}
\label{sec:cyclic}

Although our LTA definition is general and allows cycles, our
treatment thus far has deliberately excluded cycles for ease of
elucidation. However, the underlying refinement-typed language
$\lambda_{\S{lta}}$ often contain terms whose specifications require
both general parametric polymorphism (types depending on other
(non-refined types)), as well as abstract
refinements~\cite{abstractref}, i.e., refinement types parameterized
with refinement predicates.  Precisely capturing these specifications
will require cyclic edges in the LTA. For instance, cycles allow us to
model the space where a type variable $\alpha$ can be instantiated
with any type (both base and refined), and a refinement variable
$\phi$ to any well-formed refinement formula. This is shown
graphically in Fig.~\ref{fig:cycles}. \\
\begin{wrapfigure}{r}{.50\textwidth}
\includegraphics[scale=.55]{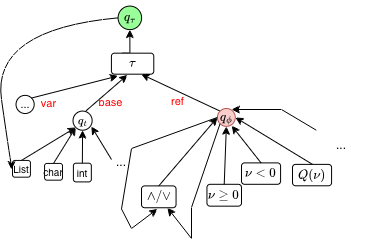}
\caption{Cycles in LTA for polymorphic types and abstract refinements. States in cycles are colored.}
 \label{fig:cycles}
\vspace*{-.3in}
\end{wrapfigure}

The state $q_t$ captures all
valid base-refinement types, where a polymorphic list constructor
\S{List} can be instantiated with any possible type (the cyclic edge
back to $q_t$). The refinement formula in $\tau$ in turn can be
instantiated with all valid refinement predicates captured by another
state $q_{\phi}$, with standard predicates like \{ $\nu \geq 0, \nu <
0$, etc.,\}, Method predicates ($Q(\nu)$), as well as conjunction and
disjunction of other predicates (cyclic edges from $q_{\phi}$).
 
Such cycles effectively represent an unbounded sequence of refinements
and types characterized by position constraints. Thus, any transition
constraint referring to a state/transition that is part of a cycle
would result in a logical formula that duplicates and accumulates
predicates for every potential iteration, making the formula unwieldy
for the solvers.  To maintain tractability, we impose a structural
restriction, conceptually similar to earlier works~\cite{ecta}, on the
LTA: the position constraints associated with any transition are
prohibited from referring to a state that participates in a cycle.  We
formally define this using a \textit{Dependency Graph} for an LTA, and
when are LTA constraints \textit{Well-formed}:
The acyclic constraint restriction enables us to prevent the infinite unrolling of logical obligations, ensuring that the resulting SMT queries are well-behaved, while still supporting the polymorphic types and abstract refinements in specifications.
Additionally, the above structural restriction on constraints allows us to directly extend our semantic intersection and pruning strategies to cyclic forms~\footnote{Additional details about the Cyclic LTA can be found in an accompanying Appendix}.

\section{Soundness and Completeness}
\label{sec:proofs}
For a given upper bound {\it k} on the size of programs being
synthesized, the {\sc LTASynthesize} algorithm is both sound and
complete assuming the
validity of each library function against their
specifications.\footnote{Proofs can be found can be found in the accompanying Appendix.}


\begin{theorem}[Soundness]
Given a type environment $\Gamma$ that relates library functions $f_i = \lambda (\overline{x_{i,j}}). e_{f_i}$ with their refinement types
$f_i : \overline{(x_{i,j} : \tau_{i,j})} \rightarrow \{ \nu :  t_i \mid \phi_{i} \} \in \Gamma$, and a synthesis query  $\varphi = \overline{(y_i : \tau_i)} \rightarrow \{\nu : t \mid \phi \}$, 
if {\sc LTASynthesize} ($\Gamma, \varphi, \S{k}$) =($\mathcal{A}_{\S{min}}$, \S{Terms} = $\mathrm{\{ }$e $\mid$ e $\in$ $\denotation{\mathcal{A}_{\S{min}}}$ $\mathrm{\}}$), then $\forall e \in \S{Terms}, \Gamma \vdash e: \varphi$, where $\Gamma$ is consistent with $\mathcal{A}_{\S{min}}$. 
\end{theorem}

\begin{theorem}[Completeness]
\label{thm:completeness}
    Given a type environment $\Gamma$ that relates library functions $f_i = \lambda (\overline{x_{i,j}}). e_{f_i}$ with their refinement types
$f_i : \overline{(x_{i,j} : \tau_{i,j})} \rightarrow \{ \nu :  t_i \mid \phi_{i} \} \in \Gamma$, and a synthesis query  $\varphi = \overline{(y_i : \tau_i)} \rightarrow \{\nu : t \mid \phi \}$, 
if {\sc LTASynthesize} ($\Gamma, \varphi, \S{k}$) = $\bot$, then $\nexists$
    a term $e \in \denotation{\mathcal{A}_{\S{complete}}}$ containing fewer than $k+1$ library function calls, such that
    $\Gamma \vdash$ $e$ : $\varphi$ and $\Gamma$ is consistent with $\mathcal{A}_{\S{complete}}$. Where $\mathcal{A}_{\S{complete}}$ is the complete LTA of size $k$, for the given $\Gamma$, generated without any reduction.
\end{theorem}


\section{Evaluation}
\label{sec:eval}
Our evaluation considers the following three research questions:
\begin{itemize}
\item[\textbf{RQ1}] \textbf{Effectiveness of \name{}}:
How effectively does \name{}(shown as He) synthesize programs satisfying complex refinement-type specifications, compared to other specification-guided, component-based synthesis tools?

\item[\textbf{RQ2}] \textbf{Scalability with query complexity}:
As synthesis queries increase in size and control-flow complexity, how well does \name{} scale in terms of success rate and synthesis time?

\item[\textbf{RQ3}] \textbf{Impact of LTA reduction strategies}:
What is the impact of pruning and similarity-based reductions on LTA size, search-space reduction, and overall synthesis efficiency?
\end{itemize}

\subsection{RQ1: Effectiveness and Comparison with Other tools}
\paragraph{Benchmarks}
To address RQ1, We evaluate \name{} on 14 benchmark queries drawn from Hoogle+~\cite{hoogleplus} and Hectare~\cite{ecta},
re-implemented in OCaml and classified into first order, higher order,
and polymorphic categories. Each query is refined using three distinct
refinement specifications, yielding 42 total benchmarks~\footnote{
A detailed description of our benchmark set and evaluation results can be found in the accompanying Appendix.}.
We evaluate \name{} against Hoogle+ and Synquid~\cite{synquid}, a deductive synthesis tool based on refinement types. A detailed description of the benchmark set is provided in the supplemental material.

\paragraph{Results}
Fig.~\ref{fig:scatter1} summarizes synthesis time and LTA minimization statistics for select benchmarks (20/42) where at least one other tool succeeded. \name{} successfully synthesizes solutions for all 42 benchmarks, with a maximum runtime of 11 seconds and an average runtime of 7.6 seconds. In contrast, Hoogle+ solves only 6 benchmarks and is approximately 6x slower on those instances. Synquid solves 20 benchmarks and is on average 4.5x slower than \name{} on commonly solved queries. Hectare is omitted as it does not support refinement based specifications. 

22.5\% of the total synthesis time on the RQ1 benchmarks was spent on solving SMT queries. Note that Hegel improves overall SMT query cost by reducing the number of such queries needed for synthesis compared to enumerative techniques, leading to substantial improvement in synthesis times.

Across benchmarks, LTA minimization substantially reduces the automaton size, with an average reduction of 73\%, directly contributing to improved synthesis performance.
All experiments were conducted with a timeout of 3 minutes and a maximum synthesis bound of five library function calls,


\begin{figure}[t!]
\centering 
\includegraphics[width=1.1\textwidth]{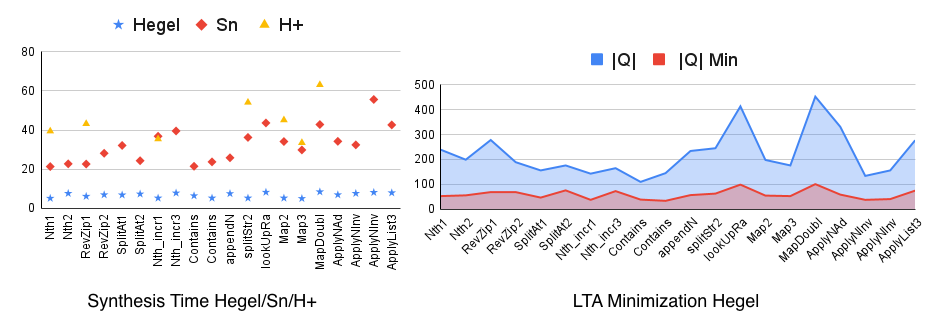}
\caption{Synthesis Results for selected RQ1 Benchmarks where at least one of the baselines tools succeeded. Synthesis time in seconds (Left) and LTA Minimization, showing \# states without LTA reductions (|Q|) vs with LTA reductions (|Q| Min) (Right). }
\label{fig:scatter1}
\end{figure}

\subsection{RQ2, Scaling \name{} to larger and complex queries}
{\it Benchmarks}.
To evaluate scalability, we consider eight synthesis queries adapted from verification benchmarks involving stateful database applications~\cite{database-examples}. These benchmarks require longer call sequences, explicit state threading, and nontrivial control flow. All benchmarks are expressed using monadic state passing to enable compositional synthesis.\\
{\it Results}.
\name{} successfully synthesizes solutions for all RQ2 benchmarks, with synthesis times ranging from 22.7 seconds to slightly over one minute. \textit{The average size of for these solutions is around 19 function calls (maximum being 33) with around 4 control flow branches on average.} 
Synquid solves only the smallest benchmarks, requiring over two minutes in those cases, and times out on all larger instances. Hoogle+ fails to solve any benchmark within the timeout of 6 minutes.
Across these benchmarks, the average number of LTA states constructed is approximately 1200. Pruning and minimization reduce the state space by roughly 70 percent, enabling \name{} to scale to substantially more complex synthesis tasks than existing tools.


\subsection{RQ3: Impact of irrelevant code pruning and similarity reduction}
To evaluate the impact of LTA reductions, we conduct ablation experiments using three restricted variants of \name{}: one disabling pruning, one disabling similarity reduction, and a baseline disabling both. We compare variants in terms of total synthesis time and the number of enumerated program terms, relative to the baseline.
For RQ1 benchmarks, disabling either pruning or similarity reduction increases synthesis time by 2-3x, while disabling both causes failures on nearly half of the benchmarks. For RQ2 benchmarks, pruning alone is insufficient to ensure scalability, with the variant disabling similarity reduction failing on 3 out of 9 queries. Across both benchmark categories, the full system achieves the smallest search space, while restricted variants enumerate 2-4.5x more terms without corresponding gains in solvability.
These results demonstrate that both pruning and similarity based reductions are necessary to achieve efficient and scalable refinement typed component based synthesis.

\section{Related Work}
\label{sec:related}

\noindent 
{\bf Component-based Synthesis.} There is a long line of work on the
use of CBS in the context of domain-specific
languages~\cite{table-synthesis,oracle-guided-synthesis} as well as
general-purpose programming
domains~\cite{sypet,tygus,cdcl-synthesis,rbsyn,frangel,viser,cobalt-tech}.
Unlike us, most prior CBS approaches rely on base types~\cite{tygus,sypet} or limited effect information~\cite{rbsyn}.
Our contributions in this paper extend prior
work by enabling CBS to be applied when specifications and queries are
equipped with logical refinements, substantially increasing the
complexity of the search process. \\
{\bf Using similarity/equality information for search.}  
E-graphs~\cite{Egg} allows efficient reduction of an enumeration space similar to the motivation
underlying our approach. Equality saturation has been applied to
enable efficient abstraction learning~\cite{babble}
inductive synthesis~\cite{egg-synthesis} and program
analysis~\cite{egg-analysis}.
These approaches rely on syntactic equivalence and do not directly support
semantic similarity notions based on refinement subtyping or logical entailment.

User provided logical equivalences have also been proposed~\cite{manual} to accelerate
synthesis, but do not scale to large component libraries.
Our idea of similarity reduction is also related to the notion of
\textit{observational equivalence} found in programming-by-example
synthesis aproaches~\cite{Albarghouthi2013,Miltner2022,Dillig23}.
These techniques compare synthesized programs on a given set of
inputs and prune the search-space in a bottom-up, inductive
synthesis setting. Its inherent unsoundness makes this mechanism
infeasible for specification-guided synthesis.\\
{\bf Tree automata for program synthesis.} 
Tree automata have been used to compactly represent large spaces of programs~\cite{dace,ecta,reactive,Reduction}
in synthesis and verification as discussed in earlier sections.
However, these representations cannot express semantic relationships induced
by refinement typing rules.

Blaze~\cite{blaze} introduces \emph{abstract finite tree automata} (AFTA), whose states like those of LTAs, are defined by a type and an associated predicate. However, the two models differ fundamentally.
First, LTA state annotations are strictly more expressive. AFTA predicates are quantifier-free formulas over an abstract domain~\cite{neilson}, whereas LTAs admit quantified formulas over program variables, enabling a direct encoding of refinement typing semantics and the use of refinement subtyping to guide search and reduction.
Second, unlike AFTA, LTAs support function types, higher-order programs, and let-bindings, leading to substantially different synthesis algorithms.
CTA~\cite{cata} extends tree automata with logical transition guards from decidable theories, analogous to symbolic finite automata(SFA)~\cite{sfa}. While this enables relational acceptance checking, CTA guards (true for SFA in general) cannot relate sub-automata as in LTA or ECTA, limiting their ability to capture typing and subtyping semantics.
LTAs are strictly more general than ECTA since any ECTA can be represented as an LTA that only uses equality constraints. 
This can be done by replacing ternary transitions like $\tau$ (var, base, ref) (see Figure~\ref{fig:baseqta}) with a unary one $\tau$ (base). The operation is similar to a standard type \textit{erasure} operation~\cite{liquidextended}. \\
{\bf Refinement types and conflict-driven learning for synthesis.}
Refinement types have been used previously to guide program synthesis~\cite{cobalt-tech,synquid,ecta},
most notably in deductive synthesis systems such as Synquid~\cite{synquid}.
These approaches focus on deriving programs that satisfy refinement
specifications, but do not address scalability in component based synthesis
with large libraries as compared extensively in our evaluations.
Other synthesis techniques based on conflict driven learning~\cite{cobalt-tech,sypet} avoid re-exploring
failing programs, but do not identify semantic equivalences among non failing
candidates.
Our approach complements refinement based synthesis by introducing a compact
representation and similarity based reduction mechanism that scales to complex
component based queries.




\section{Conclusions}
\label{sec:conc}


This paper presents a new component-based synthesis algorithm and tool, \name{}, for libraries and queries equipped with refinement-type specifications. Such specifications induce sparse solution spaces, rendering naive enumerative synthesis ineffective. We address this challenge using a novel tree automata variant, Liquid Tree Automata (LTA), which compactly represents the space of well-typed programs, supports efficient construction, and enables semantics-driven pruning via refinement-based similarity. Our evaluation shows that \name{} scales to complex refinement-typed queries beyond the reach of existing synthesis techniques.


%
%
%
\bibliographystyle{plain}
\bibliography{paper}
\section{Appendix/Supplemental Material for the Main Paper}
\label{sec:supplementary}

\pagebreak


\subsection{A Detailed Comparison of FTA, VSA, ECTA and LTA with examples.}

Several data structures could in principle be used to represent large spaces of candidate programs, including
Version Space Algebras~\cite{vsa},
e-graphs~\cite{egg-synthesis} and Finite Tree Automata
(FTA)~\cite{tata}. In particular, FTA have been shown to be
effective in representing the space of untyped programs, satisfying a set of input-output
examples~\cite{Dillig23}, as well as simply-typed programs~\cite{ecta,tata}.
However, these representations are insufficient for synthesis under refined library and query specifications, where correctness depends on enforcing logical implication and semantic relationships between subterms. 

\NEW{The reason is that these data structures have limited  capability at best to relate sub-programs in their structure. A typical example which standard FTA and VSA fail to recognize is the set of terms \{ $f (t , t)\, |\,  t \in \{ f_1, f_2, f_3 \}$ \}, where $f_i$ are symbols in the language of FTA. Note that to correctly represent this space, the automata must relate the two subtrees representing the two arguments to $f$.
Figure~\ref{fig:ap-comparison} shows (from left to right) the VSA and FTA for the unconstrained space \{ $f (t_1 , t_2) | t_i \in \{ f_1, f_2, f_3 \}$ \}. A VSA has two kind of nodes, a \textit{union node} (\fbox{$\bigcup$}), representing a union of all its children, and a \textit{join node} (\fbox{$\bowtie$}), representing a function application to all the terms represented by its children. An FTA can also be understood in a similar fashion with \textit{states} (\textcircled{$q$}) as analogous to the VSA \textit{union nodes} and \textit{transitions} (\fbox{.}) similar to the \textit{join nodes}. Note that in the figure both the FTA and VSA fail to restrict both sub-trees for $f$ to be identical and thus accept all possible terms like \{ $f (f_i, f_j) \mid i, j \in \{ 1, 2, 3 \} \}.$ }

\begin{figure}[t]
\includegraphics[scale=.50]{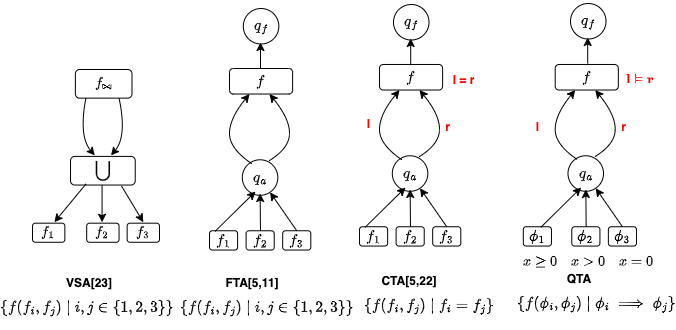}
 \caption{A comparison for representation of space of terms using VSA, FTA, ECTA and LTA.}
 \label{fig:ap-comparison}
\end{figure}
To represent such constrained spaces over trees, \textit{constrained FTA} have been explored~\cite{tata,ecta} allowing syntactic equality and dis-equality constraints between sub-terms. For instance, the CTA (third figure in Figure~\ref{fig:ap-comparison}) shows such a constrained FTA for the above example, where both sub-trees are constrained to be identical using a constraint on the transition (\textcolor{red}{{\sf l = r}}), where \textcolor{red}{l} and \textcolor{red}{r} are variables capturing paths in the automata. Consequently, the FTA will only accept terms of the form $\{f (f_i, f_j) \mid f_i = f_j$\}

Unfortunately, these syntactically constrained automata too fall short when attempting to capture the space of programs with refined logical specifications, of the kind we need for the refined variant of our motivating example. 
For instance, consider a slightly modified example of terms \{ $f (t_1 , t_2) | t_1, t_2 \in \{ \phi_1, \phi_2, \phi_3 \}$ $\wedge \ t_1 \implies t_2$ \} with $\phi_i$ being logical predicates from a decidable logic fragment. The CTA definition is lacking both in allowing such formulas in the structure as well as constraining the space to accept these logically constrained terms. 
The challenge is primarily because, a) extended FTA cannot allow logical predicates as symbols in the automata as is the case here. b) The constraints in CTA allow only syntactic comparison between sub-terms, while the task at hand requires capturing logical or semantic relations. 
Such terms arise naturally in domains like program semantics, deductive reasoning and richer type systems, particularly capturing the sub-typing relations in these richer types.


\subsection{syntax and Semantics for \calculus}
\label{sec:syntax}
\begin{figure}[h]
\small
\centering 
\small \begin{tabular}{l l l l} 
\textsf{x}, $f$ $\in$ \emph{Variables}\\
$\alpha \in$ \emph{Type Variables}\\
$Q \in$ \emph{Refinement Variables}\\
d $\in$ \emph{Constructors}& ::= & () $\mid$ \S{true} $\mid$ \S{false} $\mid$ \S{O} $\mid$ \S{S} $\mid$ \S{Cons} $\mid$ \S{Nil} \\
\emph{c} $\in$ \emph{Constants}& ::= & $\mathbb{B}$ $\mid$ $\mathbb{N}$ $\mid$ $\mathbb{Z}$ $\mid$ \ldots $\mid$ d\ $\overline{c}$ \\
{\sf v} $\in$ \emph{Value} & ::= & \emph{c} $\mid$  $\lambda$ (\textsf{x}:$\tau$). \eip $\mid$ \S{x} $\mid$ $\Lambda$ ($\alpha$ : $\kappa$). \eip $\mid$ $\Lambda$ (\S{Q} : \S{t}). \eip  \\
%
$\eip$ $\in$ \emph{Expression} & ::= & \textsf{v} $\mid$ {\bf let} \ \S{x} = \eip \ {\bf in} \ \eip{}   $\mid$ {\bf if} \S{v} \ {\bf then} \eip\ {\bf else} \eip\ $\mid$ \S{v} \ \S{v} $\mid$ \S{v} \ [t] $\mid$ \S{v} \ [$Q$] \\
{\sf t} $\in$ \emph{Base-Type} & ::= & {\sf int} $\mid$ {\sf bool} $\mid$ \ldots $\mid$ t list $\mid$ t tree {\sf t} \ldots $\mid$ {\sf t} $\rightarrow$ {\sf t} $\mid$ $\alpha$ \\
$\tau$ $\in$ \emph{Type} & ::= & \{$\nu$ : {\sf t} | $\phi$ \} $\mid$ (\textsf{x} : $\tau$) $\rightarrow$ $\tau$    \\
$\sigma \in$ \emph{Type Schema} &  ::= & $\tau$ $\mid$ $\forall \alpha : t$. $\sigma$ $\mid$ $\forall$ (\S{Q} : t). $\sigma$\\

%
%
$\phi \in$ \emph{Qualifiers} & ::= & \textsf{true} $\mid$ \textsf{false} $\mid$ $\kappa$
$Q(\overline{x_i})$ $\mid$ $\neg$ $\phi$ $\mid$ $\phi$ $\wedge$ $\phi$  $\mid$ $\phi \lor \phi$  $\mid$ $\phi$ $\Rightarrow$ $\phi$ $\mid$ $\forall$ (x : t). $\phi$ $\mid$ $\exists$ (x : t). $\phi$ 
\\

$\Gamma$ $\in$ \emph{Type Context} & ::= & $\varnothing$ $\mid$ $\Gamma$, x : $\tau$ $\mid$ $\Gamma$, $\phi$ $\mid$ $\Gamma$, $g$ : $\tau$ $\mid$ $\Gamma, \alpha : \kappa$ $\mid$ $\Gamma$, $Q$ : t\\
\hline
\end{tabular}%
\caption{Extended $\lambda_{\S{LTA}}$ Expressions and Types with parametric polymorphism and abstract refinement variables} 
\label{fig:ap-syntax} 
\end{figure}

\pagebreak

\begin{figure*}[htbp]
\begin{flushleft}
\bigskip
{\bf Expression Typing}\quad \fbox{\small $\Gamma \vdash_{\mathcal{A}}$ $\eip$ : $\sigma$}
\end{flushleft}
\begin{minipage}{0.1\textwidth}
\begin{center}
  \inference[{\sc T-var}]{\mathcal{R} (\mathcal{A}, \Gamma) & \Gamma (\S{x}) = \sigma}
                         {\Gamma \vdash_{\mathcal{A}} \S{x} : \sigma}
\end{center}
\end{minipage}
\begin{minipage}{0.3\textwidth}
\begin{center}
\inference[{\sc T-fun}]{ \mathcal{R} (\mathcal{A}, \Gamma) & \Gamma \cup (\S{x} : \tau_1) \vdash_{\mathcal{A}} \eip : \tau_2}
		            {\Gamma \vdash_{\mathcal{A}} \lambda (\S{x : \tau_1}).\eip\, : \tau_1 \rightarrow \tau_2}
\end{center}  
\end{minipage}\\[5pt]
\bigskip
\begin{minipage}{0.3\textwidth}
\begin{center}
\inference[{\sc T-App}]{\mathcal{R} (\mathcal{A}, \Gamma) \\
\Gamma \vdash_{\mathcal{A}} \eip_1 : (\S{x} : \tau_1 \rightarrow \tau_2) & \Gamma \vdash_{\mathcal{A}} \eip_2 : \tau_1} 
		    {\Gamma \vdash_{\mathcal{A}} \eip_1 \eip_2 : \tau_2[\eip_2 / \S{x}]}
\end{center}  
\end{minipage}
\bigskip
\begin{minipage}{0.4\textwidth}
\begin{center}

\inference[{\sc T-If}]{ 
\Gamma \vdash_{\mathcal{A}} \eip : \{ \nu : \S{bool} \mid \phi \} & \Gamma, \phi \vdash_{\mathcal{A}} \eip_t : \tau \\
\Gamma, \neg \phi \vdash_{\mathcal{A}} \eip_f : \tau & \mathcal{R} (\mathcal{A}, \Gamma)} 
		    {\Gamma \vdash_{\mathcal{A}} \mathsf{ if \ (\eip) \ then \ \eip_t \ else \ \eip_f} : \tau}
\end{center}  
\end{minipage}
\begin{minipage}{0.4\textwidth}
\inference[{\sc T-Subtype}]{\Gamma \vdash_{\mathcal{A}} \eip : \sigma_1 & \Gamma \vdash_{\mathcal{A}} \sigma_1 <: \sigma_2 \\ \mathcal{R} (\mathcal{A}, \Gamma)
}
		    {\Gamma \vdash_{\mathcal{A}} \eip : \sigma_2}
\end{minipage}\\[5pt]
\bigskip
\begin{minipage}{0.4\textwidth}
\inference[{\sc T-Gen-T}]{\mathcal{R} (\mathcal{A}, \Gamma) & \Gamma \vdash_{\mathcal{A}} \eip : \sigma & \alpha \notin FV (\Gamma)}
		    {\Gamma \vdash_{\mathcal{A}} \Lambda \alpha. \eip : \forall \alpha. \sigma}
\end{minipage}
\begin{minipage}{0.4\textwidth}
\inference[{\sc T-Let}]{\mathcal{R} (\mathcal{A}, \Gamma) & \Gamma \vdash_{\mathcal{A}} \eip : \tau_1 \\ 
\Gamma \cup (\S{x} : \tau_1) \vdash_{\mathcal{A}} \eip_1 : \tau}
		    {\Gamma \vdash_{\mathcal{A}} {\bf let} \ \S{x} = \eip \ {\bf in} \ \eip{} : \tau}
\end{minipage}\\[5pt]
\bigskip

\begin{minipage}{0.3\textwidth}
\inference[{\sc T-Inst-T}]{ \mathcal{R} (\mathcal{A}, \Gamma) \\ 
\Gamma \vdash_{\mathcal{A}} \eip : \forall \alpha. \sigma &  \Gamma {\vdash^{\S{wf}}}_{\mathcal{A}} \S{t}}
		    {\Gamma \vdash_{\mathcal{A}} \eip \ [t] : \sigma [t/\alpha]}
\end{minipage}
\begin{minipage}{0.4\textwidth}
\inference[{\sc T-Gen-P}]{\mathcal{R} (\mathcal{A}, \Gamma) \\
                    \Gamma \cup (\S{Q} : t) \vdash_{\mathcal{A}} \eip : \sigma }
		    {\Gamma \vdash_{\mathcal{A}} \Lambda \S{Q} : \S{t}. \eip : \forall \S{Q} : \S{t}.\sigma}
\end{minipage}\\[5pt]
\bigskip
\begin{minipage}{0.35\textwidth}
\inference[{\sc T-Inst-P}]{\mathcal{R} (\mathcal{A}, \Gamma) \\
    \Gamma \vdash_{\mathcal{A}} \eip : \forall \S{Q} : \S{t}.\sigma  & \Gamma \vdash_{\mathcal{A}} \phi : \S{t}}
		    {\Gamma \vdash_{\mathcal{A}} \eip \ [\phi] : \sigma[\S{P}/\S{Q}]}
\end{minipage}

\caption{Extended Typing Semantics for $\lambda_{\S{LTA}}$ Expressions}
\label{fig:ap-e-typing-expressions}
\end{figure*}

\begin{figure*}[htbp]
\begin{flushleft}
\bigskip
{\bf Well-Formedness}\quad \fbox{\small $\Gamma {\vdash^{\S{wf}}}_{\mathcal{A}}$ $\tau$}
\end{flushleft}
\bigskip
\begin{minipage}{0.35\textwidth}
\inference[{\sc WF-Base}]{ \mathcal{R} (\mathcal{A}, \Gamma) & 
\Gamma \equiv \overline{x_i : \tau_i} & \forall i. \Gamma {\vdash^{\S{wf}}}_{\mathcal{A}} \tau_i \\
\Gamma \cup \nu : t {\vdash^{\S{wf}}}_{\mathcal{A}} \phi}
{\Gamma {\vdash^{\S{wf}}}_{\mathcal{A}} \{ \nu : t \mid \phi \}}
\end{minipage}
\begin{minipage}{0.4\textwidth}
\inference[{\sc WF-ABS-E}]{\mathcal{R} (\mathcal{A}, \Gamma) \\
\Gamma \cup (x : \tau_1) {\vdash^{\S{wf}}}_{\mathcal{A}} \tau_2}
{\Gamma {\vdash^{\S{wf}}}_{\mathcal{A}} (x : \tau_1) \rightarrow \tau_2}
\end{minipage}\\[5pt]
\begin{minipage}{0.35\textwidth}
\inference[{\sc WF-Prim}]{\mathcal{R} (\mathcal{A}, \Gamma) & \S{t} \in T}{\Gamma {\vdash^{\S{wf}}}_{\mathcal{A}} \S{t}}
\end{minipage}
\begin{minipage}{0.4\textwidth}
\inference[{\sc WF-pred}]{\mathcal{R} (\mathcal{A}, \Gamma) & \phi \in \Phi}
{\Gamma {\vdash^{\S{wf}}}_{\mathcal{A}} \phi}
\end{minipage}\\[5pt]

\begin{minipage}{0.4\textwidth}
\inference[{\sc WF-Var}]{\mathcal{R} (\mathcal{A}, \Gamma) & x \in Vars}
{\Gamma {\vdash^{\S{wf}}}_{\mathcal{A}} \ x}
\end{minipage}\\[5pt]

\begin{minipage}{0.4\textwidth}
\inference[{\sc WF-ABS-$\alpha$}]{\mathcal{R} (\mathcal{A}, \Gamma) & \Gamma \cup \alpha : \kappa {\vdash^{\S{wf}}}_{\mathcal{A}} \tau}
{\Gamma {\vdash^{\S{wf}}}_{\mathcal{A}} \ \forall \alpha. \tau}
\end{minipage}\\[5pt]

\begin{minipage}{0.4\textwidth}
\inference[{\sc WF-ABS-$Q$}]{\mathcal{R} (\mathcal{A}, \Gamma) & \Gamma \cup Q : \S{t} {\vdash^{\S{wf}}}_{\mathcal{A}} \tau}
{\Gamma {\vdash^{\S{wf}}}_{\mathcal{A}} \ \forall Q : \S{t}. \tau}
\end{minipage}\\[5pt]

\begin{flushleft}
\bigskip
{\bf Subtyping}\quad \fbox{\small $\Gamma \vdash_{\mathcal{A}}$ $\sigma_1$ <: $\sigma_2$}
\end{flushleft}
\bigskip

 \inference[{\sc T-Sub-Base}]{ \mathcal{R} (\mathcal{A}, \Gamma) & \Gamma {\vdash^{\S{wf}}}_{\mathcal{A}} \{ \nu : \S{t} \mid \phi_1 \} \\ \Gamma {\vdash^{\S{wf}}}_{\mathcal{A}} \{ \nu : \S{t} 
\mid \phi_2 \} & 
			  \Gamma \vDash \phi_1 => \phi_2 }
				{ 
				\Gamma \vdash_{\mathcal{A}} \{ \nu : \S{t} \mid \phi_1 \} <: \{ \nu : \S{t} \mid \phi_2 
\}
				  } 

\bigskip

 \inference[{\sc T-Sub-Arrow}]{ \mathcal{R} (\mathcal{A}, \Gamma) & \Gamma \vdash_{\mathcal{A}} \tau_{21} <: \tau_{11} &\Gamma \vdash_{\mathcal{A}} \tau_{12} <: \tau_{22}}
			  { \Gamma \vdash_{\mathcal{A}} (\S{x }: \tau_{11}) \rightarrow \tau_{12}  <: (\S{x} : \tau_{21}) 
\rightarrow \tau_{22} } 

\bigskip
\begin{minipage}{0.40\textwidth}
 \inference[{\sc T-Sub-TVar}]{ \mathcal{R} (\mathcal{A}, \Gamma) & \Gamma \vdash_{\mathcal{A}} \sigma_1 <: \sigma_2}
			  { \Gamma \vdash_{\mathcal{A}} \forall \alpha. \sigma_1  <: \forall \alpha. \sigma_2 } 
\end{minipage}

\bigskip

\begin{minipage}{0.40\textwidth}
 \inference[{\sc T-Sub-PVar}]{ \mathcal{R} (\mathcal{A}, \Gamma) & \Gamma, (Q : t) \vdash_{\mathcal{A}} \sigma_1 <: \sigma_2}
			  { \Gamma \vdash_{\mathcal{A}} \forall Q : t. \sigma_1  <: \forall Q : t. \sigma_2 } 
\end{minipage}

\caption{Extended Typing Semantics for well-formedness and subtyping of \calculus terms.}
\label{fig:ap-e-typing}
\end{figure*}

\pagebreak


\subsection{\automata{} Construction: Extended Rules}
\begin{figure*}[htbp]
\begin{flushleft}
\bigskip
{\bf Well-formedness}\quad \fbox{\footnotesize
    $\begin{array}{c} 
    \mathcal{F}, \mathcal{A} \vdash^{\S{wf}} f \in \mathcal{F} (\overline{q_i}) \T{\psi} q      
    \end{array}$
}
\end{flushleft}
\bigskip
\begin{minipage}{0.3\textwidth}
{\footnotesize
\begin{center}
\inference[{\sc wf-prim}]{\S{t} \in \S{T}_{\mathcal{F}} & q_{\S{t}} \notin Q}{\mathcal{F}, \mathcal{A} \vdash^{\S{wf}} \S{t} () \T{} q_{\S{t}}}
\end{center}
}
\end{minipage}
\begin{minipage}{0.3\textwidth}
{\footnotesize
\begin{center}
\inference[{\sc wf-pred}]{\phi \in \Phi_{\mathcal{F}} & q_{\phi} \notin Q}{\mathcal{F}, \mathcal{A} \vdash^{\S{wf}} \phi \T{} q_{\phi}}
\end{center}
}
\end{minipage}
\begin{minipage}{0.3\textwidth}
{\footnotesize
\begin{center}
\inference[{\sc wf-var}]{x \in \S{Vars}_{\mathcal{F}} & q_{x} \notin Q}{\mathcal{F}, \mathcal{A} \vdash^{\S{wf}} x \T{} q_{x}}
\end{center}
}
\end{minipage}
\bigskip
\bigskip
\begin{minipage}{0.4\textwidth}
{\footnotesize
\begin{center}
\inference[{\sc wf-base}]{(\tau \equiv \{ x : t \mid \phi \}) \in \tau_{\mathcal{F}} & q_{\tau} \notin Q} 
{\mathcal{F}, \mathcal{A} \vdash^{\S{wf}} \tau (q_x, q_t, q_{\phi}) \T{} q_{\tau}}
\end{center}
}
\end{minipage}
\begin{minipage}{0.4\textwidth}
{\footnotesize
\begin{center}
\inference[{\sc wf-arrow}]{(\tau_{\rightarrow} \equiv \tau_i \rightarrow \tau_j) \in \tau_{\mathcal{F}} \\
\tau_i, \tau_j \in \tau_{\mathcal{F}}  & q_{\tau_{\rightarrow}} \notin Q} 
{\mathcal{F}, \mathcal{A} \vdash^{\S{wf}} \tau_{\tau{\rightarrow}} (q_{\tau_i}, q_{\tau_j}) \T{} q_{\tau_{\rightarrow}}}
\end{center}
}
\end{minipage}
\begin{minipage}{0.4\textwidth}
{\footnotesize
\begin{center}
\inference[{\sc wf-t-abs}]{q_{\alpha}, q_{\tau} \in Q \\
\psi = q_{\S{tabs}} \blacktriangleright \S{tvar}.\S{type} = q_{\S{tabs}} \blacktriangleright \S{type.base}} 
{\mathcal{F}, \mathcal{A} \vdash^{\S{wf}} \S{tabs}(q_{\alpha}, q_{\tau}) \T{\psi} q_{\S{tabs}}}
\end{center}
}
\end{minipage}
\begin{minipage}{0.4\textwidth}
{\footnotesize
\begin{center}
\inference[{\sc Q-goal}]{\varphi = \overline{(\mathsf{x_i} : \tau_i)} \rightarrow \tau & q_{\S{term}_k}, q_{\tau} \in Q \\ q_{\S{goal}} \in Q_f \\
\psi = \textsc{SubType} (q_{\S{term}_k}\blacktriangleright{\S{type}}, q_{\S{goal}}\blacktriangleright{\S{type}})
}
{\mathcal{F}, \mathcal{A}  \vdash^{\S{wf}} \S{goal} (q_{\tau}, q_{\S{term}_k}) \T{\psi} q_{\S{goal}}}
\end{center}
}
\end{minipage}

\begin{flushleft}
{\bf Transitions}\quad \fbox{\footnotesize
    $\begin{array}{c} \mathcal{F}, \mathcal{A} \vdash  f \in \mathcal{F} (\overline{q_i}) \T{\psi} q \\ 
    \end{array}$
}
\end{flushleft}
\bigskip
\begin{minipage}{0.25\textwidth}
{\footnotesize
\begin{center}
\inference[{\sc e-$\alpha$}]{  \S{Fresh}(\alpha) & q_{\alpha} \notin Q}
{\mathcal{F}, \mathcal{A} \vdash \alpha () \T{} q_{\alpha}}
\end{center}
}
\end{minipage}
\begin{minipage}{0.25\textwidth}
{\footnotesize
\begin{center}
\inference[{\sc e-$\kappa$}]{  \S{Fresh}(\kappa) & q_{\kappa} \notin Q}
{\mathcal{F}, \mathcal{A} \vdash \kappa () \T{} q_{\kappa}}
\end{center}
}
\end{minipage}
\begin{minipage}{0.25\textwidth}
{\footnotesize
\begin{center}
\inference[{\sc e-var}]{ x : \tau \in \mathcal{F} & q_x \notin Q}
{\mathcal{F}, \mathcal{A} \vdash x (q_{\tau}) \T{} q_x}
\end{center}
}
\end{minipage}

\bigskip
\begin{minipage}{0.20\textwidth}
{\footnotesize
\begin{center}
\inference[{\sc e-$\tau$-shape}]{\kappa () \T{} q_{\kappa}, 
\alpha () \T{} q_{\alpha} \in \Delta \\
\nu () \T{} q_{\nu} \in \Delta \\
q_{\tau} \notin Q}
{\mathcal{F}, \mathcal{A} \vdash \tau (q_{\nu}, q_{\alpha}, q_{\kappa}) \T{} q_{\tau}}
\end{center}
}
\end{minipage}
\begin{minipage}{0.4\textwidth}
{\footnotesize
\begin{center}
\inference[{\sc e-app}]{q_f, q_a \in Q & \tau (\overline{q_i}) \T{} q_{\tau} \in \Delta\\
\psi = \textsc{SubType} (q_{f}\blacktriangleright{\S{out}}, q_{\S{app}}\blacktriangleright{\S{type}}) \wedge \\
\ \ \ \  \theta = [q_{\S{a}}\blacktriangleright{\epsilon}/q_{f}\blacktriangleright{\S{in}}] \\
\theta. \textcolor{red}{(}\textsc{SubType} (q_{\S{a}}\blacktriangleright{\S{type}}, q_{f}\blacktriangleright{\S{in}}) \textcolor{red}{)}}
{ \mathcal{F}, \mathcal{A} \vdash \S{app} \ (q_{\tau}, q_f, q_a) \T{\psi} q_{\S{app}}}
\end{center}
}
\end{minipage}
\bigskip
\bigskip
\bigskip
\vspace*{.05in}

\begin{minipage}{0.4\textwidth}
{\footnotesize
\begin{center}
\inference[{\sc e-const}]{ \vdash c : \tau \in \mathcal{F} \\ q_c \notin Q}
{\mathcal{F}, \mathcal{A} \vdash c (q_{\tau}) \T{} q_c}
\end{center}
}
\end{minipage}
\begin{minipage}{0.4\textwidth}
{\footnotesize
\begin{center}
\inference[{\sc e-if}]{q_b, q_t, q_f \in Q & \tau (\overline{q_i}) \T{} q_{\tau} \in \Delta \\
\psi =  ((q_b \blacktriangleright {\S{ref}}) \wedge \textsc{SubType} (q_{t}\blacktriangleright {\S{type}}, q_{\S{if}}\blacktriangleright{\S{type}})) \wedge \\
\   \   \   \  (\neg (q_b \blacktriangleright{\S{ref}}) \wedge \textsc{SubType} (q_{f}\blacktriangleright{\S{type}}, q_{\S{if}}\blacktriangleright{\S{type}}))}{ \mathcal{F}, \mathcal{A} \vdash \S{if} \ (q_{\tau}, q_b, q_t, q_f) \T{\psi} q_{\S{if}}}
\end{center}
}
\end{minipage}

\begin{minipage}{0.4\textwidth}
{\footnotesize
\begin{center}
\inference[{\sc e-infer}]{\delta = f(q_{\tau}, \overline{q_i}) \xhookrightarrow{\psi} q & 
\delta_{\tau} \in \delta \blacktriangleright{\tau} = (q_{\nu}, q_{\alpha}, q_{\kappa}) \T{} q_{\tau} \\
\mathcal{R} (\mathcal{A}, \Gamma) & 
\S{Solve} (\Gamma, \psi, \denotation{q_{\alpha}}, \denotation{q_{\kappa}}) = (M_{\kappa}, M_{\alpha})
}
{\mathcal{F}, \mathcal{A} \vdash (M_{\kappa}, M_{\alpha}). \delta_{\tau}}
\end{center}
}
\end{minipage}
\bigskip
{\footnotesize
\begin{equation}
\label{eq:theta}
\denotation{\theta = [\delta \blacktriangleright p / \delta \blacktriangleright p']} = 
            [ \S{x} / \S{y} \mid \delta_{p} \in \delta\blacktriangleright p, \delta_{p'} \in \delta\blacktriangleright p' ,\S{y} = \S{Symbol} (\delta_{p}), \S{x} = \S{Symbol} (\delta_{p'})]
\end{equation}
}

{\footnotesize
}
\caption{Rules for constructing transitions $\Delta$,  basis for {\sc WF} and {\sc Transition.} with rules {\sc E-$\alpha$},  {\sc E-$\kappa$} and {\sc E-infer}}
\label{fig:ap-typing}
\end{figure*}

\pagebreak

\subsection{LTA Reductions}

\begin{figure*}[htbp]
\begin{flushleft}

{\bf Pruning}\quad\fbox{\footnotesize
     $\begin{array}{c} 
        \mathcal{A} \vdash \Delta \rightsquigarrow \Delta' \ \ \ \  \mid \ \ \ \ \mathcal{A} \vdash  \delta \rightsquigarrow^{\psi_a} \delta' 
      \end{array}$
           
} 

\end{flushleft}
\bigskip
{\footnotesize
\begin{minipage}{0.30\textwidth}
\begin{center}
\inference[{\sc p-trans}]{ \delta \equiv f (\ldots) \T{\psi_j \wedge \ldots} q & 
\psi_j \equiv \theta. p_1 (= \mid \vDash) p_2 \\ 
\Delta_{r} = \{ \delta_r \mid \mathcal{A} \vdash (\delta_{p_1, i} \in \delta\blacktriangleright p_1)  \rightsquigarrow^{\psi_j} \delta_r \} 
} 
{\mathcal{A} \vdash \Delta \rightsquigarrow \Delta[\delta\blacktriangleright p / \Delta_r]} 
\end{center}
\end{minipage}
%
\begin{minipage}{0.40\textwidth}
\begin{center}
\inference{\delta \equiv f (\ldots) \T{\psi_j \wedge \ldots} q & 
\psi_j \equiv p_1 = p_2 \\
\delta_{p_1, i} \in \delta\blacktriangleright p_1 & \delta_{p_2, j} \in \delta\blacktriangleright p_2 \\
\delta_{\S{r}} = {\sqcap}_{\S{Syntax}} (\delta_{p_1, i}, \delta_{p_2, j})
} 
 { \mathcal{A} \vdash  \delta_{p_1, i}  \rightsquigarrow^{\psi_j} \delta_{\S{r}}}[{\footnotesize{\sc p-syn-eq}}]   
\end{center} 
\end{minipage}


\begin{minipage}{0.40\textwidth}
\begin{center}
\inference{\delta \equiv f (\ldots) \T{\psi_j \wedge \ldots} q & 
\psi_j \equiv \theta. p_1 \vDash p_2 \\
\delta_{p_1, i} \in \delta\blacktriangleright p_1 & \delta_{p_2, j} \in \delta\blacktriangleright p_2 \\
\delta_{\S{r}} = {{\sqcap}^{\psi_j}}_{\S{Semantics}}(\delta_{p_1, i}, \delta_{p_2, j})
} 
 { \mathcal{A} \vdash  \delta_{p_1, i}  \rightsquigarrow^{\psi_j} \delta_{\S{r}}}[{\footnotesize{\sc p-sem-ent}}]  
\end{center} 
\end{minipage}
}
\bigskip

{\footnotesize
\begin{equation}
\label{eq:sem}
{{\sqcap}^{\psi}}_{\S{Semantics}}(\delta_i, \delta_j) = 
        \begin{cases}
             \delta_i & \sigma (\psi_j) = \theta. \star \vDash \star, \S{Symbol} (\delta_i) = \phi_i, \S{Symbol} (\delta_j) = \phi_j,  \\
                       & \S{Given \ the \ automaton} \ \mathcal{A} , \ s.t. \ \mathcal{R} (\Gamma, \mathcal{A}) \\
                       & \denotation{\Gamma} \wedge \denotation{\theta}. \denotation{\phi_i} \implies \denotation{\phi_j} \\
            \delta_{\bot} & \S{otherwise}
        \end{cases}
\end{equation}
}

\begin{flushleft}
{\bf Similarity}\quad \fbox{\footnotesize
  $\begin{array}{c} 
\mathcal{A} \vdash^{\S{\bf sim}} \delta_i \lesssim \delta_j \ \ \ \  \mid \ \ \ \ \mathcal{A} \vdash \mathcal{E} \rightsquigarrow \mathcal{E}'
\end{array}$
}
\end{flushleft}
\bigskip
\begin{minipage}{0.4\textwidth}
{\footnotesize
 \inference[\footnotesize {\sc s-trans}]{ \psi_{<:} = \textsc{SubType} \textcolor{red}{(} \delta_i\blacktriangleright \S{type}, \delta_j\blacktriangleright \S{type} \textcolor{red}{)} \\  \delta_{\bot} \notin {{\sqcap}^{\psi_{<:}}}_{\S{Semantics}} (\delta_i\blacktriangleright \S{type} , \delta_j\blacktriangleright \S{type}) }
			  { \mathcal{A} \vdash^{\S{{\bf sim}}} \delta_i \lesssim \delta_j} 
}
\end{minipage}
\begin{minipage}{0.40\textwidth}
{\footnotesize
\inference[{\sc s-eq}]{ (\delta_i, \delta_j) \notin \mathcal{E} \\
\mathcal{A} \vdash^{\mathbf{sim}} \delta_i \lesssim \delta_j } {\mathcal{A} \vdash \mathcal{E} \rightsquigarrow \mathcal{E} \cup \{(\delta_i, \delta_j)\}}
}
\end{minipage}

\bigskip

\begin{flushleft}

{\bf Minimization}\quad\fbox{\footnotesize
     $\begin{array}{c} 
        (\mathcal{A}, \mathcal{E})  \vdash \Delta \rightsquigarrow \Delta' \ \ \ \  \mid \ \ \ \ \vdash (\mathcal{A}, \mathcal{E})   \rightsquigarrow (\mathcal{A}', \mathcal{E}') 
      \end{array}$
           
} 

\end{flushleft}
\bigskip
{\footnotesize
\begin{minipage}{0.40\textwidth}
\begin{center}
\inference[{\sc m-trans}]{\delta_i, \delta_j \in \Delta & (\delta_i , \delta_j) \in \mathcal{E} \\
\delta_i \equiv f (q_1, q_2, \ldots q_j \ldots q_n) \T{\psi_i} q_i  \\
\delta_j \equiv f' (q_1', q_2', \ldots q_m') \T{\psi_j} q_j  \\
\Delta' = \bigcup_k \textcolor{red}{\{} \delta_k[q_j \mapsto q_i] \\ 
\ \ \ \ \ \ \ \ \ \ \ \  \  \   \    \ \mid \delta_k = \hat{f} (\hat{q_1}, {\bf q_j}, \ldots \hat{q_m} ) \hookrightarrow \hat{q_k}\textcolor{red}{\}} } 
{(\mathcal{A}, \mathcal{E}) \vdash \Delta \rightsquigarrow (\Delta \cup \Delta') \setminus 
\{\delta_j\}} 
\end{center}
\end{minipage}
\bigskip
\begin{minipage}{0.40\textwidth}
\begin{center}
\inference{ {\footnotesize \mathcal{A} \equiv (Q, \mathcal{F}, Q_f, \Delta}) & {\footnotesize  (\mathcal{A}, \mathcal{E}) \vdash \Delta \rightsquigarrow^{\star} \Delta'}} 
 { \vdash (\mathcal{A}, \mathcal{E}) \rightsquigarrow ((Q, \mathcal{F}, Q_f, \Delta'), \varnothing)}[{\footnotesize{\sc m-LTA}}]   
\end{center} 
\end{minipage}
}
\caption{Complete set of Similarity inference and LTA Minimization.}
\label{fig:ap-similarity}
\end{figure*}

\subsubsection{Example}
\begin{figure}[htbp]
\includegraphics[scale=.55]{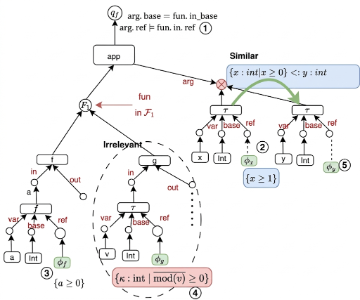}
\caption{A LTA for the example library. The green arrow relates similar sets of transitions.}
 \label{fig:ap-reductions}
\vspace*{-.4in}
\end{figure}


\pagebreak


To illustrate the irrelevant code pruning defined above we use the following example.
\begin{example}
\label{ex:reductions}
Consider a space of valid refinement-typed terms for a typing
environment given as follows: \textsf{[ \{ x : int $\mid$ x $\geq$ 0
    \}, y : int, f : (a : \{$\nu$ : int $\mid$ a $\geq$ 0\})
    $\rightarrow$ \{$\nu$ : int $\mid$ (a == 0) $\implies$ $\nu$
    $\geq$ a $\wedge$ ((a > 0) $\implies$ $\nu$ $\leq$ a\}), g : \{
    $\nu$ : int $\mid$ \S{mod} ($\nu$) $\geq$ 0\} $\rightarrow$
    bool]}.  Figure~\ref{fig:ap-reductions} shows a partial LTA for the
space of terms of sizes up to 1 function call, constructed using the
rules in Figure~\ref{fig:typing} for the given typing environment.
\end{example}
The irrelevant portion of the LTA is shown in the dashed block with label \textcircled{4}
The {\sc p-trans} rule when applied on the
transition {\sf app} (call it $\delta_{\S{app}}$) updates this
transition using constraints on $\delta_{\S{app}}$.  The constraint
has two conjuncts; here, we just consider the semantic entailment constraint
($\psi_j$ = {\sf arg.ref $\vDash$ fun.in.ref}) ((\textcircled{1})),
shown in green.  Rule {\sc
  p-Sym-ent} thus applies and reduces it to $\delta_{\S{app}}
\rightsquigarrow^{\psi_j = \S{arg.ref} \vDash \S{fun.in.ref}}
\delta_{\S{app}}'$.  This rule calculates
$\delta_{\S{app}}\blacktriangleright \S{arg.ref}$, which in the figure
is $\phi_x$ (\textcircled{2}) and $\delta_{\S{app}}\blacktriangleright
\S{fun.in.ref}$, which includes two set of transitions, $\phi_f$
((\textcircled{3})) and $\phi_g$ ((\textcircled{4})).  Finally the
$\sqcap_{\S{Semantics}}$ ($\phi_x$, $\phi_g$) is $\delta_{\bot}$, as
the check, ($\forall \nu. \nu \geq 1 \vDash \S{mod} (\nu) \geq 0$)
does not hold. On the other hand, $\sqcap_{\S{Semantics}}$ ($\phi_x$,
$\phi_f$) is non-$\delta_{\bot}$.  Thus, $\phi_g$((\textcircled{4}))
is reduced to $\delta_{\bot}$, which upon normalization, leads to
reduction of whole transition corresponding to function \S{g}, shown
in red as irrelevant.

To illustrate, the similarity based minimization, consider Figure~\ref{fig:ap-reductions} again, using {\sc
  s-trans} rules between the transitions for \S{x} and \S{y}, shown in
the box with label {\bf Similar}, the $\psi_{<:}$ constraint is shown
in the blue box. Since, the constraint holds (under variable renaming),
these two transitions are marked as similar (shown by the
green arrow) and added to $\mathcal{E}$. The {\sc m-trans} rule finally
uses this similarity information to remove the transition for \S{y}
while keeping \S{x}.

\pagebreak

\subsection{Synthesis Details and Implementation}

\NNEW{In the LTA construction and pruning rules and the {\sc LTASynthesize} algorithm above, for ease of illustration, we have abstracted away several details about how we maintain variable scoping, infer types for transitions, and do efficient term extraction. Below, we discuss some of these along with several other details. We use the following example to illustrate these details. 
}
\NNEW{For illustration, we will consider a library $\mathcal{F}$ = [\S{f} : (l : \{ $\nu$ : [a] | len ($\nu$) > 0 \}) $\rightarrow$ \{ $\nu$ : [a] | len($\nu$) = len(l) \};
xs : \{ $\nu$ : [int] | len ($\nu$) = 1\} ; ys : [char]; 
g : (l : [char]) $\rightarrow$ [char]].
Figure~\ref{fig:ap-terms} shows a portion of minimized LTA for terms of size two with transition \fbox{\S{app}} for this library.}

\begin{wrapfigure}{r}{.50\textwidth}
\includegraphics[scale=.450]{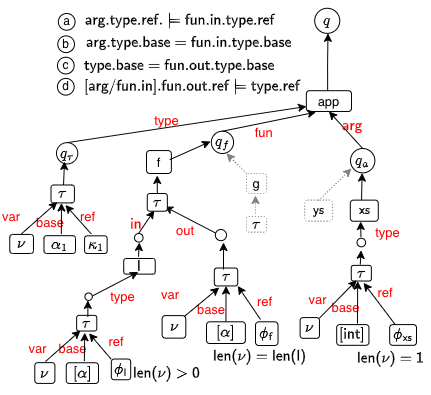}
\caption{ An Example LTA, portions shaded out for elucidation.} 
\label{fig:ap-terms}
\end{wrapfigure}

\subsubsection{Variable Scoping and Typing Environment in LTA}
\NNEW{The details of variable scoping during LTA construction and pruning are important to understand how enumeration works with refinement types in a LTA. 
To simplify scoping decisions and manage the typing environment, we made several design choices.
First, we require terms in our synthesis language \calculus to be in A-normal form, and to have a unique binding variable $\S{t}_i$ for each application and conditional term. 
We also refactor each library function specification, alpha-renaming all argument variables . This allows us to avoid unwanted variable capture across library functions without explicitly keeping track of scope information about bound and free variables.}

\NNEW{Additionally, to build the typing environment, each term binding variable in ANF  is ascribed a set of possible types that can be associated with it. To implement this structurally in LTA, we extend each n-ary expression transition (e.g.\S{app} transition) in the LTA, to (n+1) arity with an additional incoming edge for the \textcolor{red}{type} of the resulting expression (e.g., function application term). See, for example, the arrows with label \textcolor{red}{type} in Figure~\ref{fig:ap-terms} for the \S{app} transition that has an edge ($q_{\tau} \rightarrow$ \fbox{\S{app}}),
$q_{\tau}$ to represent a set of valid types that can be ascribed to the application term (described next). 
Finally, we build a global typing environment mapping each expression type (annotated and inferred) with the binding variable using a typing environment construction function. This function is derived from a relation relating LTAs to typing environments; details are provided in the next section. }

\begin{figure*}[htbp]
\begin{flushleft}
\bigskip
{\bf Typing Environment Relation for $\mathcal{A}$}\quad \fbox{\footnotesize
    $\begin{array}{c} 
    \mathcal{R}_{Q} (q, \Gamma) \\
    \mathcal{R}_{\Delta} (\delta, \Gamma) \\
    \mathcal{R} (\mathcal{A}, \Gamma) 
    \end{array}$
}
\end{flushleft}
\bigskip
\begin{minipage}{0.3\textwidth}
{\footnotesize
\begin{center}
\inference[{\sc $\mathcal{R}_{Q}$-Var}]{\delta \equiv x (q_{\tau}) \T{\psi} q}{\mathcal{R}_{Q} (q, \{(x : \denotationNode{q_{\tau}})\})}
\end{center}
}
\end{minipage}
\begin{minipage}{0.3\textwidth}
{\footnotesize
\begin{center}
\inference[{\sc $\mathcal{R}_{Q}$-Const}]{\delta \equiv c (q_{\tau}) \T{\psi} q}{\mathcal{R}_{Q} (q, (c : \denotationNode{q_{\tau}}))}
\end{center}
}
\end{minipage}
\begin{minipage}{0.3\textwidth}
{\footnotesize
\begin{center}
\inference[{\sc $\mathcal{R}_{Q}$-GEN}]{ \delta \equiv f ({q}_1, {q}_2,\ldots,{q}_n, q_{\tau}) \T{\psi} q \\
\mathcal{R}_{Q} (q, \Gamma)  &  \mathcal{R}_{\Delta} (\delta, \Gamma')}{\mathcal{R}_{Q} (q, (\Gamma \cup \Gamma'))}
\end{center}
}
\end{minipage}
\bigskip
\bigskip    
\begin{minipage}{0.5\textwidth}
{\footnotesize
\begin{center}
\inference[{\sc $\mathcal{R}_{\Delta}$-GEN}]{ \delta \equiv f ({q}_1, {q}_2,\ldots,{q}_n, q_{\tau}) \T{\psi} q \\
\S{Fresh} (v_{\delta}) & i \in [1 \ldots n]. \ \mathcal{R}_{Q} (q_i, \Gamma_i)}{\mathcal{R}_{\Delta} (\delta, (\bigcup_{i \in [1 \ldots n]} \Gamma_i) \cup \{(v_{\delta} : \denotationEdge{f.\S{type}})\})}
\end{center}
}
\end{minipage}
\begin{minipage}{0.3\textwidth}
{\footnotesize
\begin{center}
\inference[{\sc $\mathcal{R}$-EMPTY}]{}{\mathcal{R} (\mathcal{A}, \varnothing)}
\end{center}
}
\end{minipage}
\bigskip
\begin{minipage}{0.4\textwidth}
{\footnotesize
\begin{center}
\inference[{\sc $\mathcal{R}$-FINAL}]{q \in Q_f &  \mathcal{R} (\mathcal{A}, \Gamma) \\
\mathcal{R}_{Q} (q, \Gamma')} {\mathcal{R} (\mathcal{A}, \Gamma \cup \Gamma')}
\end{center}
}
\end{minipage}
\begin{minipage}{0.4\textwidth}
{\footnotesize
\begin{center}
\inference[{\sc $\mathcal{R}$-TERM}]{q \in Q & \mathcal{R} (\mathcal{A}, \Gamma) \\
\nexists (f ({q}_1, {q}_2,\ldots,{q}_n, q_{\tau}) \T{\psi} q) \in \Delta & \mathcal{R}_{Q} (q, \Gamma')} {\mathcal{R} (\mathcal{A}, \Gamma \cup \Gamma')}
\end{center}
}
\end{minipage}

\caption{Inductively defined Relation $\mathcal{R}$, capturing the relation between $\mathcal{A}$ and the Typing Environment $\Gamma$. A little abuse of notation for presentation, denotation for a type node (and edge) $\denotationNode{q_{\tau}}$ is a singleton set, and we use it to also denote the element in the set in Rules {\sc $\mathcal{R}_{Q}$-Var}, {\sc $\mathcal{R}_{Q}$-Const} and {\sc $\mathcal{R}_{\Delta}$-GEN}}
\label{fig:environment}
\end{figure*}

The inference rules in Fig 1 define a relation $\mathcal{R} (\mathcal{A}, \Gamma)$, that holds iff we can construct $\Gamma$ using $\mathcal{A}$. 
The rules are constructive, giving a way to construct $\Gamma$ from $\mathcal{A}$ using two other relations -- a state-environment relation $\mathcal{R}_{Q} (q, \Gamma)$, capturing the construction of an environment from a sub-automaton rooted at a state $q$; and a transition-environment relation $\mathcal{R}_{\Delta} (\delta, \Gamma)$, capturing environment construction from a given transition $\delta$.

{\sc $\mathcal{R}$-EMPTY} is the base rule, which says that $\mathcal{R} (\mathcal{A}, \varnothing)$  holds for any automaton $\mathcal{A}$. Operationally, it means we can construct an empty environment from any automaton $\mathcal{A}$.
Rules {\sc $\mathcal{R}_{Q}$-Var} construct an environment with a singleton pair, mapping a variable $x$ to the denotation of the state $q_{\tau}$, which is a singleton set containing the type for $x$ (see Fig. 7 in the paper for definition). Intuitively, we are adding each type-annotated variable pair ($x, \tau$) to the environment.
{\sc $\mathcal{R}_{Q}$-Const} does a similar construction for type annotated constants, using corresponding transitions.

{\sc $\mathcal{R}_{Q}$-GEN} is a generic state-environment rule applicable to any transition target state $q$. It says, given a constructed environment $\Gamma$ from $q$ (possibly constructed by some other transition) and an environment  $\Gamma'$, constructed from this transition, we can take the union of  $\Gamma'$ and $\Gamma$ to construct the new environment from $q$.  

 {\sc $\mathcal{R}_{\Delta}$-GEN} similarly is a generic transition-environment construction rule. It says, given environments ($\Gamma_i$) constructed from each of the incoming states $q_i$, we can construct a new extended environment from the current transition by taking a union of $\Gamma_i$ and adding a new pair for the type of the current transition (refereed by the position $f$.type), mapped to a fresh variable for the transition $v_f$. 

Finally, {\sc $\mathcal{R}$-TERM} and {\sc $\mathcal{R}$-Final} construct a typing environment for the automata by adding the environments for each terminal states (states with no outgoing transition) and final states.

Collectively, these rules allow us to construct a typing environment corresponding to any given automata $\mathcal{A}$.

\pagebreak

\subsection{Resolving Refinement Predicates for \textcolor{red}{type} Edges in Transitions}
\NNEW{LTA construction and pruning also rely on inferring a set of feasible types for each transition. For instance see the incoming state ($q_{\tau} \rightarrow$ \fbox{\S{app}}) in the example. The possible set of types is shown using a transition (\fbox{$\{ \nu : \alpha_1 \mid \kappa_1\}$}
$\rightarrow$ \textcircled{$q_{\tau}$}), 
constructed using \textbf{Transitions} construction rules {\sc E-$\alpha$}, {\sc E-$\kappa$} and {\sc E-$\tau$-shape}  in Figure~\ref{fig:typing}}.

Inferring this type set precisely requires inferring the base type for type variable $\alpha_1$ and the refinement predicate (here $\kappa_1$). 
We earlier skipped the principal rule {\sc E-infer} for inferring possible values for these variables.
The {\sc E-infer} rules fetches each \textcolor{red}{type} transition for a given transition $\delta$ using, construct a typing environment consistent with the automata $\mathcal{A}$, using $\mathcal{R}$ and uses a auxiliary constraint solving procedure \S{Solve} to solve the constraint over these typing and refinement variables in $\psi$ thus infering corresponding mappings $M_{\alpha}$ and $M_{\kappa}$.
Inferring $M_{\alpha}$ is relatively straightforward using 
syntactic comparison between terms at constrained location, e.g., the given constraint relating function and argument type structure allows us to infer in this case that the type of the application term is [\S{int}].
Inferring $M_{\kappa}$ is more convoluted and needs some elucidation.
We use the transition constraints $\psi$ to generate logical \textit{implication} relations over these refinement variables.

For instance, consider semantic constraint \textcircled{d}\textcircled{d}, $\psi_j$ = [\S{arg}/\S{fun.in}]. \S{fun.out.ref} $\vDash$ \S{type.ref}.
To translate this to constraints over variables, we get the symbols at these positions and constrtaint these symbols according to $\psi_j$. For instance,  $\delta \blacktriangleright$\S{fun.in} = \{ \S{l} \}, $\delta \blacktriangleright$\S{arg} = \{\S{xs, ys} \}, however, \S{ys} we can ignore for our example as it does not satisfy \textcircled{a} and thus will be pruned out. Thus we will have $\theta$ = [\S{xs}/\S{l}]. 
Next $\delta \blacktriangleright$\S{fun.out.ref} = \{ $\phi_{\S{f}}$ \} and $\delta \blacktriangleright$\S{type.ref} = \{ $\kappa_1$\}.

Now relating these sets using $\psi$ we get a constraint:
(i) \S{len} ($\nu$) = \S{len (l)} $\vDash$ $\kappa_1$. 
Similarly \textcircled{b} will give us another constraint (ii) \S{len} ($\nu$) = 1 $\vDash$ \S{len} ($\nu$) > 0.
Now, to solve these constraints; a) we perform these checks in a consistent typing environment $\Gamma$ mapping variables occurring free in these constraint with types. For instance, for our example we have (\S{xs} : \{ $\nu$ : [int] $\mid$ \S{len}($\nu$) = 1\}) $\in \Gamma$.
b) We must lift these $\vDash$ checks in the logical implication checks using our earlier discussed interpretation $\denotation{.}$.
This along with the substitution in \textcircled{d} gives us the following logical constraint. 
$\denotation{\S{xs} : \{ \nu : [int] \mid \S{len}(\nu) = 1 \}}$ 
$\wedge$ ($\forall \nu$. \S{len} ($\nu$) = 1 $\implies$ \S{len} ($\nu$) > 0) 

$\wedge$ [\S{xs}/\S{l}] 

(\S{len} ($\nu$) = \S{len (l)} $\implies$ $\kappa_1$).

Using the $\denotation{}$ semantics, this translates to:
(\S{len}(\S{xs}) = 1) 
$\wedge$ true $\wedge$ [\S{xs}/\S{l}] (\S{len} ($\nu$) = \S{len (xs)} 
$\implies$ $\kappa_1$). 
Thus giving us a possible values of $M_{\kappa}$ as [$\kappa_1 \mapsto$ \S{len} ($\nu$) = 1].

\subsection{Handling Cycles}

\begin{definition}[Dependency Graph and Cyclic States]
Given a Liquid Tree Automaton $\mathcal{A} = (Q, F, Q_f, \Delta)$, we define its \textbf{dependency graph} as a directed graph $G_\Delta = (Q, E)$, where an edge $(q', q) \in E$ exists if there exists a transition rule $f(q_1, \dots, q_n) \xrightarrow{\psi} q \in \Delta$ such that $q' = q_i$ for some $1 \le i \le n$. 

A state $q \in Q$ is \textbf{cyclic}, 
denoted $q \in Q_{cyc}$, if there exists a non-empty path from $q$ to itself in $G_\Delta$ ($q \to^{+} q$). 
Otherwise, the state is acyclic ($q \in Q_{acyc}$).
\end{definition}

\begin{definition}[Well-formed LTA Constraints ]
A constrained transition $t = f(q_1, \dots, q_n) \xrightarrow{\psi} q \in \Delta$ is \textbf{well-formed} if and only if every position referenced within its constraint $\psi$ resolves to an acyclic state. 
Formally, let $pos \in \{1, \dots, n\}$ be a position index referenced in $\psi$ in a transition $\delta$. We require:
$\forall p \in \text{positions}(\psi) \implies \delta  \blacktriangleright p  \in Q_{acyc}$.

An LTA $\mathcal{A}$ satisfies the structural acyclic constraint restriction if all transitions $t \in \Delta$ are well-formed.
\end{definition}


\pagebreak
\subsection{Details of Soundness and Completeness Theorems and Proofs}

For a given upper bound {\it k} on the size of programs being
synthesized, the  {\sc LTASynthesize} algorithm is both sound and
complete.

\paragraph{\bf Soundness}.

\begin{theorem}[Soundness]
Given a type environment $\Gamma$ that relates library functions $f_i = \lambda (\overline{x_{i,j}}). e_{f_i}$ with their refinement types
$f_i : \overline{(x_{i,j} : \tau_{i,j})} \rightarrow \{ \nu :  t_i \mid \phi_{i} \} \in \Gamma$, and a synthesis query  $\varphi = \overline{(y_i : \tau_i)} \rightarrow \{\nu : t \mid \phi \}$, 
if {\sc LTASynthesize} ($\Gamma, \varphi, \S{k}$) =($\mathcal{A}_{\S{min}}$, \S{Terms} = $\mathrm{\{ }$e $\mid$ e $\in$ $\denotation{\mathcal{A}_{\S{min}}}$ $\mathrm{\}}$), then $\forall e \in \S{Terms}, \Gamma \vdash e: \varphi$, where $\Gamma$ is consistent with $\mathcal{A}_{\S{min}}$. 
\end{theorem}

Informally Stated: Programs synthesized by the {\sc LTASynthesize} procedure are correct with
respect to the provided query specification $\Psi$ assuming the
validity of each library function against their
specifications.

We assume that every library has a correct type annotation using the following annotation correctness assumption.
\begin{proposition}[Annotation correctness]
\label{assume:type-correctness}
 $\forall$ $g \in \mathcal{F}$, and the query $\varphi$ = $\overline{(x_i : \tau_i)}$ $\rightarrow$ $\tau$, 
if $g$ has an annotated type $\tau$, then $\mathcal{F}, \overline{ x_i : \tau_i}$ $\vdash$ $g : \tau$.
\end{proposition}

We begin with defining an important Lemma for proving the soundness theorem:

\begin{lemma}[\automaton Denotation Correctness]
\label{lemma:denotation-correctness}

    Assuming, the type annotation in the library and query arguments are correct, $\forall$ \automata $\mathcal{A}_{\S{q}}$, rooted at a node \S{q}, $\forall$ $e$ $\in \denotation{\mathcal{A}_{\S{q}}}$. $\mathcal{F}$, $\mathcal{A}$
    $\vdash$ $e : \varphi$. {\bf Where $\mathcal{F}$, $\mathcal{A}$
    $\vdash$ $e : \varphi$. is same as  $\Gamma \vdash_{\mathcal{A}}$ $e$ : $\varphi$ given  $\Gamma$ = $\mathcal{F}$ and $\mathcal{R} (\mathcal{A}, \Gamma)$}
\end{lemma}

\begin{proof}
The proof is using induction on the construction of the automaton, using automata construction rules.

\begin{itemize}
    \item  We first take the base cases for well-formedness rules.
    \begin{enumerate}
        \item Cases of well-formedness of {\sc wf-prim}, {\sc wf-pred}, {\sc wf-base} and {\sc wf-arrow}, are trivially true as these are not expressions and our typing systems assumes any well-formed type or predicate has a higher-order kind.

        \item Case {\sc E-const}:
            \begin{itemize}
                \item Given a constant $c : \tau \in \mathcal{F}$
                \item Assumption Correct Annotation, $\mathcal{F} \vdash c : \tau$
                \item The automaton constructed for $c$ is given by $c (q_{\tau}) \hookrightarrow q_c$
                \item Using $\denotation{.}$ definition. $\denotation{\mathcal{A}_{q_c}}$ = $c : \tau$.
                \item Which holds using assumption.
            \end{itemize}
        \item Case {\sc E-var}: This case is similar to {\sc e-const}.
        \item Case {\sc E-app}:
            \begin{itemize}
                \item Using IH, for each $f \in \denotation{q_{f}}$ (Henceforth, using $q_f$ for $\mathcal{A}_{q_f}$, by the definition of automata rooted at a state), $\mathcal{F}, \mathcal{A} \vdash f : q_f \blacktriangleright \S{type}$ and $a \in \denotation{q_a}$
                $\mathcal{F}, \mathcal{A} \vdash a : q_a \blacktriangleright \S{type}$
                
                \item $\denotation{q_{\S{app}}}$ = $\denotationEdge{\S{app} (q_{\tau}, q_{f}, q_{a})}$.
                \item Let us assume $f \in \denotation{q_f}$, then by IH, $\mathcal{F}, \mathcal{A} \vdash f : \tau_{\S{in}} \rightarrow \tau_{\S{out}}$ for some $\tau_{\S{in}}$ and $\tau_{\S{out}}$.
                \item Similarly $a \in \denotation{q_a}$, $\mathcal{F}, \mathcal{A} \vdash a : \tau_a$.
                \item The constraint $\psi$ in the {\sc E-app} conclusion, by construction (using the {\sc SubType} equation, enforces that the $\tau_a <: \tau_{\S{in}}$ 
                \item Using the above two premise, for each term $\S{app} (f, a) \in \denotation{q_{\S{app}}}$ we have the premise to apply Expression typing rule {\sc T-App} (Figure~\ref{fig:ap-e-typing}), with $\mathcal{F}, \mathcal{A}$ as $\Gamma$. (variable to type mapping construction for $\mathcal{A}$ is straight-forward so skipped).
                \item Using {\sc T-App}, $\S{app} (f, a) \in \denotation{q_{\S{app}}}$ $\mathcal{F}, \mathcal{A} \S{app} (f, a) \vdash \tau_{out}$.
                \item Again, $\psi$ in the {\sc E-app} ensures that $q_{\S{app}} \blacktriangleright \S{type} <: \tau_{out}$.
                \item Finally, using the standard subtyping rule $\mathcal{F}, \mathcal{A} \S{app} (f, a) \vdash \tau_{\S{app}}$, giving the required proof. 
            \end{itemize}
        \item Case for construction of type-level function application {\sc e-inst-t} is similar and uses {\sc T-inst-T} rule in place {\sc T-app} rule.
        \item Case {\sc E-if}: 
            \begin{itemize}
                \item Using IH, for each $b \in \denotation{q_{b}}$ (Henceforth, using $q_f$ for $\mathcal{A}_{q_f}$, by the definition of autoamat rooted at a state), $\mathcal{F}, \mathcal{A} \vdash b : q_b \blacktriangleright \S{type}$ and $e_t \in \denotation{q_t}$
                $\mathcal{F}, \mathcal{A} \vdash e_t : q_t \blacktriangleright \S{type}$ and $e_f \in \denotation{q_f}$
                $\mathcal{F}, \mathcal{A} \vdash e_f : q_f \blacktriangleright \S{type}$
                \item $\denotation{q_{\S{if}}}$ = $\denotationEdge{\S{if} (q_{\tau}, q_{b}, q_{t}, q_{f})}$.
                \item \item Let us assume $b \in \denotation{q_b}$, then by IH, $\mathcal{F}, \mathcal{A} \vdash b : \{ \nu : bool \mid \phi_b\}$ for some $\phi_{b}$.
                \item Similarly $e_t \in \denotation{q_t}$, $\mathcal{F}, \mathcal{A} \vdash e_t : \tau_{t}$ and for $e_f \in \denotation{q_f}$ $\mathcal{F}, \mathcal{A} \vdash e_f : \tau_{f}$
                \item  The constraint $\psi$ in the {\sc E-if} conclusion, by construction (using the {\sc SubType} equation, enforces that the $\mathcal{F}, \mathcal{A}, \phi_b \vdash \tau_{e_t} <: \tau_{\S{e_{if}}}$.
                \item  Also constraint $\psi$ in the {\sc E-if} conclusion, by construction (using the {\sc SubType} equation, enforces that the $\mathcal{F}, \mathcal{A}, \neg \phi_b \vdash \tau_{e_f} <: \tau_{\S{e_{if}}}$.
                \item \item Using the above two premise, for each term $\S{if} (b, e_t, e_f) \in \denotation{q_{\S{if}}}$ we have the premise to apply Expression typing rule {\sc T-If} (Figure~\ref{fig:ap-e-typing}), with $\mathcal{F}, \mathcal{A}$ as $\Gamma$. 
                \item Using {\sc T-if}, $\S{if} (b, e_t, e_f) \in \denotation{q_{\S{if}}}$ $\mathcal{F}, \mathcal{A} \S{if} (b, e_t, e_f) \vdash \tau_{if}$.
            \end{itemize}
        \item Case {\sc e-let}: This case has the similar argument as the {\sc e-app} and {\sc e-if}.
    \end{enumerate}
\end{itemize}
\end{proof}

\begin{lemma}[{\sc Prune} preserves Denotation Correctness]
\label{lemma:prune-denotation-correctness}

    Assuming, the type annotation in the library and query arguments are correct, $\forall$ \automata $\mathcal{A}_{\S{q}}$, rooted at a node \S{q}, $\forall$ $e$ $\in \denotation{\mathcal{A}_{\S{q}}}$. $\mathcal{F}$, $\mathcal{A}$
    $\vdash$ $e : \varphi$. {\bf Where $\mathcal{F}$, $\mathcal{A}$
    $\vdash$ $e : \varphi$. is same as  $\Gamma \vdash_{\mathcal{A}}$ $e$ : $\varphi$ given  $\Gamma$ = $\mathcal{F}$ and $\mathcal{R} (\mathcal{A}, \Gamma)$}
\end{lemma}

\begin{proof}
    
\end{proof}

\begin{lemma}[{\sc Minimize} preserves Denotation Correctness]
\label{lemma:minimize-denotation-correctness}

    Assuming, the type annotation in the library and query arguments are correct, $\forall$ \automata $\mathcal{A}_{\S{q}}$, rooted at a node \S{q}, $\forall$ $e$ $\in \denotation{\mathcal{A}_{\S{q}}}$. $\mathcal{F}$, $\mathcal{A}$
    $\vdash$ $e : \varphi$. {\bf Where $\mathcal{F}$, $\mathcal{A}$
    $\vdash$ $e : \varphi$. is same as  $\Gamma \vdash_{\mathcal{A}}$ $e$ : $\varphi$ given  $\Gamma$ = $\mathcal{F}$ and $\mathcal{R} (\mathcal{A}, \Gamma)$}
\end{lemma}

\begin{proof}
    
\end{proof}

\begin{theorem}[Soundness]
Given a type environment $\Gamma$ that relates library functions $f_i = \lambda (\overline{x_{i,j}}). e_{f_i}$ with their refinement types
$f_i : \overline{(x_{i,j} : \tau_{i,j})} \rightarrow \{ \nu :  t_i \mid \phi_{i} \} \in \Gamma$, and a synthesis query  $\varphi = \overline{(y_i : \tau_i)} \rightarrow \{\nu : t \mid \phi \}$, 
if {\sc LTASynthesize} ($\Gamma, \varphi, \S{k}$) =($\mathcal{A}_{\S{min}}$, \S{Terms} = $\mathrm{\{ }$e $\mid$ e $\in$ $\denotation{\mathcal{A}_{\S{min}}}$ $\mathrm{\}}$), then $\forall e \in \S{Terms}, \Gamma \vdash e: \varphi$, where $\Gamma$ is consistent with $\mathcal{A}_{\S{min}}$. 
\end{theorem}

The proof is based on the previous denotation correctness lemma that every term in the language of a \automaton is well typed; and the fact that the {\sc LTASynthesize} generates a LTA, where the final state(s) are target of a transition with a type $\varphi$.

\begin{proof}
    \begin{itemize}
        \item If $\mathcal{A}$ = {\sc LTASynthesize} ($\mathcal{F}, \varphi, \S{k}$), then it must be a return value at either at line 11 or line 3 in the Algorithm. 
        We consider the following two cases:
        \item Case $\mathcal{A} = \mathcal{A}_0$ at Line 3:
        \begin{enumerate}
            \item From Line 3, the initially constructed LTA $\mathcal{A}_{0}$ has a solution.
            \item Using the definition for {\sc NEmpty}, $\exists q_f \in {Q_f}_0$ such that $\denotation{q_f} \neq \varnothing$. 
            \item Now such a state $q_f$ must have been constructed by {\sc WF}, rule using the query $\varphi$, such that $q_f \blacktriangleright \S{type}$ = $\varphi$.
            \item Using the Denotional Correctness Lemma ~\ref{lemma:denotation-correctness}, $\forall e \in \denotation{q_f}$, $\mathcal{F}, \mathcal{A} \vdash e : \varphi$
        \end{enumerate}
        \item Case $\mathcal{A} = \mathcal{A}_{\S{min}}$ at Line 11: 
            This is a more involved case with two main differences from the earlier case; 
            Additional call to {\sc Transition} at line 6. Plus calls to {\sc prune}, {\sc similarity} and {\sc minimize} (lines 7-9).
            \begin{enumerate}
            \item From Line 11, the minimized LTA $\mathcal{A}_{\S{min }}$ has a solution.
            \item Let us consider the automata $\mathcal{A}$ at line 6.  
            \item Using the definition for {\sc NEmpty}, $\exists q_f \in {Q_f}$ such that $\denotation{q_f} \neq \varnothing$. 
            \item Now such a state $q_f$ must have been constructed by {\sc WF}, rule, {\sc Q-goal} for the given query $\varphi$, such that $q_f \blacktriangleright \S{type}$ = $\varphi$.
            \item Now, using the Denotaional Correctness Lemma ~\ref{lemma:denotation-correctness}, $\forall e \in \denotation{q_f}$, $\mathcal{F}, \mathcal{A} \vdash e : \varphi$.
            \item Additionally using Lemma~\ref{lemma:prune-denotation-correctness} if $\mathcal{A}_{\S{r}}$ = {\sc prune} ($\denotation{A}$) then ($\denotation{A}_{\S{r}}$) $\subseteq$ ($\denotation{A}$) and $\denotation{A}_{\S{r}}$ preserve denotation correctness.
            
            \item Similarly, from Lemma~\ref{lemma:minimize-denotation-correctness}, given $\mathcal{A}_{\S{r}}$ =  {\sc minimize} ($\denotation{A}, \mathcal{E}$), for some sound similarity relation $\mathcal{E}$, then  $\mathcal{A}_{\S{r}} \subseteq$ ($\denotation{A}$) and and $\denotation{A}_{\S{r}}$ preserve denotation correctness.
            \item Thus using the soundness of the {\sc Similarity} function,
            \item Additionally from 6 and 7 we have $\denotation{\mathcal{A}_{\S{r}}} \subseteq \denotation{\mathcal{A}}$.         
            \item Finally Using 5 and 9, we have   $\forall e \in \denotation{q_f}$, $\mathcal{F}, \mathcal{A} \vdash e : \varphi$
        \end{enumerate}
    \end{itemize}
\end{proof}

\pagebreak

\pagebreak
\subsection{Completeness}

\begin{theorem}[Completeness]
\label{thm:completeness}
    Given a type environment $\Gamma$ that relates library functions $f_i = \lambda (\overline{x_{i,j}}). e_{f_i}$ with their refinement types
$f_i : \overline{(x_{i,j} : \tau_{i,j})} \rightarrow \{ \nu :  t_i \mid \phi_{i} \} \in \Gamma$, and a synthesis query  $\varphi = \overline{(y_i : \tau_i)} \rightarrow \{\nu : t \mid \phi \}$, 
if {\sc LTASynthesize} ($\Gamma, \varphi, \S{k}$), = $\bot$, then $\nexists$
    a term $e \in \denotation{\mathcal{A}_{\S{complete}}}$ containing fewer than $k+1$ library function calls, such that
    $\Gamma \vdash$ $e$ : $\varphi$ and $\Gamma$ is consistent with $\mathcal{A}_{\S{complete}}$. Where $\mathcal{A}_{\S{complete}}$ is the complete LTA of size $k$, for the given $\Gamma$, generated without any reduction.
\end{theorem}

We begin with the definition of \emph{Small subset} relation between search spaces.
Let us define a search space of all the well-typed terms of size less than $k+1$ as $S_k$. The {\sc LTASynthesize} algorithm without reductions (i.e. {\sc prune} and {\sc minimize}) (where we naturally add all transitions and states based on $L$ and typing rules) is exhaustive over this search space. Let us call the search space of the {\sc LTASynthesize} with {\sc reductions} as ${S_k}_{\S{reducced}}$.

\begin{definition}[Small-subset]
Given a library $\mathcal{F}$ and an automaton $\mathcal{A}$. A search space $S_k$ is a finite set of all possible well-typed expressions $\mid \eip \mid \leq k$ of length upto k.
A small-subset $S_k$' $\subseteq S_k$ is a search space such that if $\exists$ a well-typed $\lambda_{\S{LTA}}$ expression $\eip \in S_k$, then $\exists \eip' \in$ $S_k$' such that if $\mathcal{F}, \mathcal{A} \vdash \eip' : \tau'$, and $\mathcal{F}, \mathcal{A} \vdash \eip : \tau$, then $\tau' <: \tau$. 
\end{definition}

Since a LTA of size $k$ is a representation of a search space of terms upto size $k$, we can generalize this definition to \emph{Small-Subset-LTA}, using the denotation function for the LTA.

\begin{definition}[Small-LTA]
Given a LTA $\mathcal{A}$, 
a LTA $\mathcal{A}$' is a small-LTA for $\mathcal{A}$, (written as $\mathcal{A}$' $\subset_{\S{small}}$ $\mathcal{A}$) if the search space $\denotation{\mathcal{A}'}$,  is a \emph{small-subset} of  the search space $\denotation{\mathcal{A}}$
\end{definition}

\begin{definition}[Small-LTA-modulo-query]
Given a LTA $\mathcal{A}$, 
a LTA $\mathcal{A}$' is a small-LTA-modulo-query for $\mathcal{A}$ and a given query $\varphi$, (written as $\mathcal{A}$' $\subset_{\S{small}}^{\varphi}$ $\mathcal{A}$) if the search space $\denotation{\mathcal{A}'}$,  is a \emph{small-subset} of  the search space $\denotation{\mathcal{A}}$, for any $\mathcal{F}, \mathcal{A} \eip : \varphi$.
\end{definition}

To prove the required goal for completeness, we go in three steps:
\begin{itemize}
    \item First we proof, that if $\exists $ $\eip$, |$\eip$| $\leq k$, such that $ \mathcal{F}, \denotation{\mathcal{A}_{\S{complete}}} \vdash \eip : \varphi$, then {\sc LTASynthesize} without {\sc Prune} and {\sc Minimize} will produce an automaton $\mathcal{A}_{\S{complete}}$, such that $\eip \in \denotation{\mathcal{A}_{\S{complete}}}$.
    \item We prove that the {\sc LTASynthesize} with {\sc Minimize} always produces an automata $\mathcal{A}$ which maintains a \textit{Small-LTA} property with respect to the complete automata.
    \item We prove that {\sc prune} ($\mathcal{A}$) always produces an automata $\mathcal{A}_{\S{r}}$, such that $\mathcal{A}_{\S{r}}$ is a \textit{small-LTA-modulo-query} for $\mathcal{A}$ modulo the query $\varphi$.
    \item Use the sub-typing judgement to prove that if there exists a solution in $\mathcal{A}_{\S{complete}}$, then there must exists a solution in $\mathcal{A}$
\end{itemize}

\begin{lemma}[LTASynthesize-NoReduction-is-Complete]
    \label{lem:synthesize-complete}
    Let us assume that both {\sc Prune} and {\sc Minimize} are identity function, thus having no effect on the transitions. Let us call this variant of {\sc LTASynthesize} as {\sc LTASynthesize (-All)}.  
    $\forall \eip$, |$\eip$| $\leq k$, Iff $\mathcal{F} \vdash \eip : \varphi$, 
    then $\eip \in \denotation{\mathcal{A}_{\S{complete}}}$, 
    where $\mathcal{A}_{\S{complete}}$ = {\sc LTASynthesize (-All)}.  ($\mathcal{F}$, $\varphi$,$ k$)
\end{lemma}
\begin{proof}
  The proof follows is by  contradiction:
  \begin{itemize}
      \item Let assume that $\exists e'$, |$e'$| $\leq k$ and $\mathcal{F} \vdash e' : \varphi$.
      \item This would mean that that there exists a typing derivation in \calculus for $e'$, this holds by the soundness of the underlying type system for \calculus (Assumption~\ref{assume:type-correctness}).
      \item Given that the transition addition rules exactly mimic the typing judgments in \calculus and in  {\sc LTASynthesize (-All)} no transition is ever removed from $\mathcal{A}_{\S{complete}}$.
      \item Thus wlog, if we assume that term $e' = f (e_1, \ldots e_n)$ then, there must be a corresponding transition rule which adds a transition $f (q_1, \ldots q_n) \T{\psi} q \in \Delta_{\S{complete}}$, with $q \blacktriangleright \S{type}$ = $\varphi$.
      \item Thus using the definition for $\denotation{\cdot}$, $e' \in \denotation{\mathcal{A}_{\S{complete}}}$.
      \item Contradicts our initial assumption, thus $\forall \eip$, |$\eip$| $\leq k$, Iff $\mathcal{F} \vdash \eip : \Psi$, 
    then $\eip \in \denotation{\mathcal{A}_{\S{complete}}}$.
  \end{itemize}

\end{proof}

\begin{lemma}[Minimize-produces-small-subset-LTA]
\label{lem:minimise-small}
The {\sc Minimize} ($\mathcal{A}, \mathcal{E}$) routine always produces an automaton $\mathcal{A}_{\S{min}}$
such that $\mathcal{A}_{\S{min}}$ $\subset_{\S{small}}$ $\mathcal{A}$.
\end{lemma}
\begin{proof}
We prove this using contradiction, let us assume that it is not the case that $\mathcal{A}_{\S{min}}$ $\subset_{\S{small}}$ $\mathcal{A}$. 
\begin{itemize}
    \item By definition of $\subset_{\S{small}}$, $\exists \eip \in \denotation{\mathcal{A}}$, with $\mathcal{F}, \mathcal{A} \vdash \eip : \tau$, such that $\nexists \eip' \in \denotation{\mathcal{A}_{\S{r}}}$, such that $\mathcal{F}, \mathcal{A} \vdash \eip' : \tau'$ and $\tau' <: \tau'$.
    \item If $\eip \in \denotation{\mathcal{A}_{\S{r}}}$, it is a trivial case as the assumption is trivially false, proving the lemma.
    \item Let us say that $\eip \notin \denotation{\mathcal{A}_{\S{r}}}$, this could have happened only when {\sc s-trans} would remove some transition $\delta_e$, such that $\eip \in \denotationEdge{\delta_e}$ or $\eip \in \denotationNode{q_e}$, where $q_e$ is a target state of $\delta_e$.

    \item In both cases however, {\sc similarity} rules will ensure that $\exists \delta_e'$ such that 
    $\delta_e' \lesssim \delta_e$.

    \item Using the definition of $\lesssim$, we must have that $\forall \eip \in \denotationEdge{\delta_e}$ $\exists \eip'$ such that $\mathcal{F}, \mathcal{A} \vdash \eip : \tau$, then $\mathcal{F}, \mathcal{A} \vdash \eip' : \tau'$ and $\tau' : \tau.$

    \item This contradicts our assumption hence, $\mathcal{A}_{\S{min}}$ $\subset_{\S{small}}$ $\mathcal{A}$.
\end{itemize}
\end{proof}

\begin{lemma}[Prune-produces-small-subset-LTA]
\label{lem:prune-small}
Given a query $\varphi$, the {\sc Prune} ($\mathcal{A}$) routine  always produces an automaton $\mathcal{A}_{\S{r}}$
such that $\mathcal{A}_{\S{r}}$ $\subset_{\S{small}}^{\varphi}$ $\mathcal{A}$.
\end{lemma}

\begin{proof}
This is by definition of {\sc Prune}, which only reduces terms which could not be part of the solution.  
\end{proof}

Finally, restating the completeness with a proof piggybacking on the above proven Lemmas.
\begin{theorem}[Completeness]
\label{thm:completeness}
    Given a type environment $\Gamma$ that relates library functions $f_i = \lambda (\overline{x_{i,j}}). e_{f_i}$ with their refinement types
$f_i : \overline{(x_{i,j} : \tau_{i,j})} \rightarrow \{ \nu :  t_i \mid \phi_{i} \} \in \Gamma$, and a synthesis query  $\varphi = \overline{(y_i : \tau_i)} \rightarrow \{\nu : t \mid \phi \}$, 
if {\sc LTASynthesize} ($\Gamma, \varphi, \S{k}$), = $\bot$, then $\nexists$
    a term $e \in \denotation{\mathcal{A}_{\S{complete}}}$ containing fewer than $k+1$ library function calls, such that
    $\Gamma \vdash$ $e$ : $\varphi$ and $\Gamma$ is consistent with $\mathcal{A}_{\S{complete}}$. Where $\mathcal{A}_{\S{complete}}$ is the complete LTA of size $k$, for the given $\Gamma$, generated without any reduction.
\end{theorem}

\begin{proof}
The proof for follows directly from the Lemma~\ref{lem:synthesize-complete}, Lemma~\ref{lem:minimise-small} and Lemma~\ref{lem:prune-small}.
\end{proof}

\subsection{Informal Argument for Termination}
{\sc LTASynthesize} (Algorithm 1) is guaranteed to terminate in a finite number of steps for a given max-depth k, and a starting minimized automaton of size m. 

The {\sc Transition} routine potentially increments the depth to m' >= m. Although {\sc Prune} and {\sc Minimize} can collapse states and transitions, and hence reduce the depth to a value less than m, we still guarantee termination by remembering pruned transitions.

{\sc Transition}, {\sc Prune}, and {\sc Minimize} are deterministic implementations of the rules in Figures 7 and 8 in the main paper. 
The implementation of these functions memoizes previously pruned transitions to avoid infinite cycles among these three procedures. Thus, together: using a) the memoization, b) given that there are finitely many library functions, and c) a finite maximum size k bounds the algorithm, we are guaranteed that Algorithm always terminates.

\pagebreak
\subsection{Evaluation}

\subsection{Benchmark Construction Example}
Consider the second benchmark \texttt{RevAppend} in Figure~\ref{tab:results-rq1}
The original Hoogle query is \texttt{[a] -> [a]}. This is represented in Hoogle+ as follows:

\begin{lstlisting}
stack exec -- hplus --disable-filter=False --json='{"query": "[a] -> [a]",  
 "inExamples": [],  
 "inArgNames": ["z"]}'
\end{lstlisting}

One natural refinement of this is to reverse the original list and append it to the original.
We can capture this using a refinement of the original query, using the I/O examples as follows in Hoogle+:

\begin{lstlisting}
stack exec -- hplus --disable-filter=False --json='{"query": "[a] -> [a]",  
 "inExamples": [{ "inputs": ["[1,2,3,4]"],  
                "output": "[1,2,3,4,4,3,2,1]"},  
                { "inputs": ["[1,3,5,6]"],  
                "output": "[1,3,5,6,6,5,3,1]"},  
                { "inputs": ["[\"abcd\"]"],  
                "output": "[\"abcddcba\"]"}],
 "inArgNames": ["z"]}'
\end{lstlisting}

Finally, we also create an analogous Hegel query, capturing the similar refinement using Hegel's specification language. We do this as follows using the \texttt{mem} and \texttt{ord} method predicates/Qualifiers.

\begin{lstlisting}
revApp : (z : [int]) -> 
    {v : [int] | \(u : int), (w : int). mem (u, v) = true => mem (u , z) /\
                        len (v) == len (z) +  len (z) /\
                        ord (u, w, z) = true => 
                        (ord (u, w, v) = true /\ ord (w, u, v) = true)
                        };
\end{lstlisting}

\subsection{Motivational Example}

Following we present the details of our main motivation example, without the library.

Following is our main Hegel/Synquid Query:

\begin{lstlisting}
goal : (x:int) 
-> (y : int) 
-> (xs : [a]) 
-> { v : ([a], [a]) | 
len (fst (v)) <= x 
/\ (len (snd (v)) $<=$ len (xs) - y \/ len (snd (v) = 0)) 
/\ \(u : a). mem (fst (v). u) = true => mem (xs, u) 
/\ \(u : a). mem (snd (v), u) = true => mem (xs, u)};
\end{lstlisting}

Following is an equivalent, refined query in Hoogle+ tool:

\begin{lstlisting}
 stack exec -- hplus --disable-filter=False --json='{"query": "Int -> Int -> [a] -> ([a], [a])", 
 "inExamples": [{ "inputs": ["1", "2", "[49, 62, 82, 54, 76]"], 
                "output": "([49],[82,54,76])"},
                { "inputs": ["2", "3", "[49, 62, 82, 54, 76]"], 
                "output": "([49,62],[54,76])"},
                { "inputs": ["3", "3", "[49, 62, 82, 54, 76]"], 
                "output": "([49, 62, 82],[54, 76])"}],
 "inArgNames": ["x", "y", "z"]}'
\end{lstlisting}

\subsection{Complete Results Table for RQ1}

{\scriptsize
\begin{figure}[H]
\begin{tabular}{| l  | l | l V{2} l | l | l | l |}
\hline
\textbf{Name} & \textbf{Original and Refined Queries Description} & \#$\wedge$/$\vee$ & \multicolumn{3}{|c|}{\textbf{Refined Time(s)/\#}} & \#S \\
\hline  
    &          &    & He & H+ & Sn &   \\
\hline
Nth1   &  Swap i and j indexes in l   &   4    & 5.1 &  39.4 & 21.3 & 4  \\  
Nth2   &  Swap i and j in reverse l & 3   & 7.6 &   &  22.7 & 5  \\  
Nth3   &  Increment values at, i, j, and swap  & 4  & 9.1 &   &  & 5   \\  
\hline
RevApp1	&  Reverse l and append to itself  & 3     &   5.1 &  49.5 &    & 3  \\  
RevApp2	&  Append l to itself and reverse   & 3 & 7.3 &  &   & 4   \\  
RevApp3  & Reverse l, append and reverse  & 4 & 5.4 &  &   &4   \\  
\hline
RevZip1  & Reverse f and, zip it with s   & 3 & 6.1 & 43.2 & 22.6 &4  \\  
RevZip2  & Reverse s and zip f with it  &  4 & 6.9 &   & 28.1 & 5   \\  
RevZip3  & Zip f and s, and reverse   &  4  & 8.5 &  &   &6   \\  
\hline
SplitAt	& Split l at index n; drop m element from second proj   & 4  & 6.8 &  & 32.1 & 5  \\  
SplitAt2 & Split l at index n; drop m element from first proj & 5  & 7.3 &  & 24.3 & 4   \\  
SplitAt3  & Split l at index n; drop (n-m) element from second & 4  & 7.8 &  &  & 5   \\  
\hline
Nth\_Incr1   & Return (n+1)th element from l  &  3  & 5.2 & 35.4 & 36.8 & 3  \\  
Nth\_Incr2	 & Return (n-1)th element from l &  4  & 7.4 & 56.3 &   & 5   \\  
Nth\_Incr3   & Return nth element from l and increment it. &  4  & 7.8 &  & 39.5 &  5   \\  
\hline
CEdge1     &  True if src and tgt both are in l& 4  & 6.4 &  & 21.5 & 4  \\ 
CEdge2    &   True if src or tgt is in l        & 4    &     5.2 &   &  23.7 & 3   \\  
CEdge3    &   True if src is, but tgt not in l& 5 & 8.6 &  &   &  5   \\  
\hline
AppendN1      & Append first n elements to l&  4& 6.0 & 47.7 &   & 4 \\  
AppendN2      & Append (n-1) elements to l& 4 & 10.4 & &    & $7^{\star}$   \\  
AppendN3      & Append last n elements to l&     4 & 7.5  & & 25.8 & 5  \\  
\hline
SplitStr1 &Split s at character c& 4 & 6.5 & &   & 5  \\  
SplitStr2 & Split s at character c and swap pair& 4  & 5.2 &  & 36.2 & 5   \\  
SplitStr3 &  Split s at c, and append c in the second proj & 6  & 7.1 & &   & 5   \\  
\hline
LookRange1 &Lookup k in the values&  4 &  8.2 &  & 43.6 & $6^{\star}$  \\  
LookRange2	& Lookup k in values, else return last value &  6 &   9.1 & &   & $7^{\star}$   \\  
LookRange3    & Lookup k in values, else return first value & 4  & 8.6 & &    & $7^{\star}$    \\  
\hline
$\S{Map1}^{\textnormal{\textdagger}}$ &Map f on l and increment each element & 5 &  7.5 &  &   & 6  \\ 
$\S{Map2}^{\textnormal{\textdagger}}$ &Map f on l and decrement each element & 6 & 5.2 & 45.1 & 34.1 & 4   \\  
$\S{Map3}^{\textnormal{\textdagger}}$ &Map f on l, increment each element, and reverse & 5 & 5.0 & 33.5 & 29.8 &  4   \\  
\hline
$\S{MapDouble1}^{\textnormal{\textdagger}}$ &Map f followed by g & 4 & 10.8 &  &   &5  \\  
$\S{MapDouble2}^{\textnormal{\textdagger}}$	& Map g followed by f  & 4  & 8.9 & &   &4  \\  
$\S{MapDouble3}^{\textnormal{\textdagger}}$ & Create composition of (f.g) and apply on l & 6  & 8.4 &  &  42.8 & 4   \\  
\hline
$\S{ApplyNAdd1}^{\textnormal{\textdagger}}$ &Apply f on m, n times and add m& 6 &7.9 &  &  & 5  \\  
$\S{ApplyNAdd2}^{\textnormal{\textdagger}}$	& Apply f on m, n+1 times and add m & 6 & 6.9 &63.1 &  34.2 & 5   \\  
$\S{ApplyNAdd3}^{\textnormal{\textdagger}}$ &Apply f on m, n+1 times and add m+1 & 6 & 9.1 & &  & 6   \\  
\hline
$\S{ApplyNInv1}^{\textnormal{\textdagger}}$ & Apply g and h on l, and filter using f& 5 & 9.5& &  & 5  \\  
$\S{ApplyNInv2}^{\textnormal{\textdagger}}$ &Filter l using f, and apply g and h&  5 & 7.6& & 32.4 & 5   \\  
$\S{ApplyNInv3}^{\textnormal{\textdagger}}$ & Apply g and h on l, filter using l, apply g and h. & 6 & 8.1 & &  55.6 & 5  \\  
\hline
$\S{ApplyList1}^{\textnormal{\textdagger}}$ &Apply head of fl on s& 4  & 12.3&  &   &$8^{\star}$ \\  
$\S{ApplyList2}^{\textnormal{\textdagger}}$ & Apply last element of fl on s& 4  & 10.5 & &    &$8^{\star}$  \\  
$\S{ApplyList3}^{\textnormal{\textdagger}}$ &Apply all $\S{f} \in \S{fl}$ in sequence.  & 5 & 7.9 &  & 42.6  & $6^{\star}$  \\  
\hline
\end{tabular}
\caption{Results for experiments with Refined Hoogle+ and ECTA benchmarks.}
\label{tab:results-rq1}
\end{figure}
}

\pagebreak

\subsection{Complete Results for RQ2}

{\scriptsize
\begin{figure}[H]
\begin{tabular}{| l | p{3cm} | >{\columncolor[gray]{0.8}}l | l | l | l | l | r | r | r | r | r |}
\hline
\textbf{Name} & Desc. & \#$\wedge$/$\vee$ &  \multicolumn{4}{|c|}{\textbf{Results Hegel}} & \multicolumn{4}{|c|}{\textbf{LTA \& SMT Stats}} & T(Sn) \\
\hline
 &  &  & T(He)  & \#C & \#B  & \#R & \#SMT & SMT(s) & $|Q|$ & $|Q|_{\min}$ &  \\
\hline
NLInsert     & Add a newsletter and user           & 6 & 22.7 & 16 & 2 & 31 & 110 & 11.12 & 779  & 212 & 126.2 \\
NLRemove     & Remove a newsletter and user        & 4 & 39.9 & 20 & 4 & 35 & 189 & 19.34 & 1201 & 372 & \_     \\
NLR\_Remove  & Read articles list and remove       & 5 & 42.8 & 19 & 4 & 21 & 213 & 18.20 & 1331 & 381 & \_     \\
NLInv        & Remove with uniqueness invariant     & 8 & 52.5 & 25 & 4 & 36 & 154 & 24.19 & 1398 & 435 & \_     \\
FWInsert     & Insert a normal device               & 4 & 31.2 & 15 & 2 & 33 & 166 & 12.34 & 945  & 298 & 124.6 \\
FWMkCentral  & Insert a central device              & 4 & 65.2 & 33 & 4 & 53 & 259 & 27.13 & 1611 & 401 &        \\
FWInsConn    & Insert a device connected to all     & 5 & 36.8 & 14 & 2 & 59 & 218 & 15.12 & 806  & 261 &        \\
FWInvert     & Invert the connections               & 4 & 33.9 & 14 & 4 & 47 & 197 & 15.70 & 1176 & 323 &        \\
FWInvertDel  & Delete, and invert connections       & 6 & 47.3 & 17 & 4 & 38 & 184 & 22.90 & 1352 & 421 &        \\
\hline
\end{tabular}

\caption{Results for tailored specification-guided synthesis benchmarks, The \#C and \#B gives the total number of function calls and branches in the synthesized solution. \#R gives the number of transitions \name{} merged during the Similarity reduction and Irrelevant code pruning phases.}
\label{tab:results-rq2}
 
\end{figure}
}

\subsection{Example Synthesis output for RQ2}
Given the query {\sf NLRRemove}, the challenge is to
synthesize a solution that maintains a specific contract associated
with each library function; these include the requirement that a) the
user must be unsubscribed before removal, b) if the user has not opted
for \textit{promotions}, the email for the user must be cleared, etc.
Figure~\ref{fig:synthesized} shows the synthesized program generated
by \name{} for this query.  Note that the solution includes a total of
19 function calls and exhibits complex control flows (4 branches).

 \begin{figure}[H]
  \vspace{-.15in}
\begin{minted}[fontsize = \scriptsize, escapeinside=&&,linenos]{ocaml}
(*nLRRemove : (n : nl) -> (u : user) -> 
  (d : {v: [nlrecord] | mem (v , n , u)}) -> 
  {v : (  f : article * s : [nlrecord]) |  
  mem (f, articles (s)) 
  &$\wedge$&  &$\neg$& nlmem (s, n, u) 
  &$\wedge$& (promotions (s, u) => email (s, u))
  }*)
fun n u d ->
    let x = read (d, n, u) in 
    let x0 = fst (x) in 
    let d0 = snd (x) in
    let d1 = confirmU (d0, n, u) in
    let x1 = promotions (d1, n, u) in 
    if (not (x1)) then 
        let subscribed = subscribed (d1, n, u) in 
        if (length (subscribed) > 0) then 
            let d2 = clear_email (d1, n, u) in 
            let d3 = unsubscribe (d2, n, u) in 
            let d4 = remove (d3, n, u) in 
            (x0, d4)
        else 
            let d5 = unsubscribe (d1, n, u) in 
            let d6 = remove (d5, n, u) in 
            (x0, d6)
            
    else 
        let d7 = unsubscribe (d1, n, u) in 
        let d8 = remove (d7, n, u) in 
        (x0, d8)
\end{minted}
\caption{Synthesized Program for NLR\_Remove}
\label{fig:synthesized}
\vspace*{-.5in}
\end{figure}
\pagebreak

\subsection{RQ3: Impact of irrelevant code pruning and similarity reduction}
Because RQ3 cuts across both set of benchmarks, we perform several
ablation experiments over the queries described in the
Figures~\ref{tab:results-rq1} and ~\ref{tab:results-rq2}.  We create
three variants of \name{}, \textit{viz.} (i) \textsf{Hegel(-P)}, a
LTA-based synthesis implementation without the irrelevant code
reduction (i.e. comment out the {\sc Prune} call at line 7 in
{\sc LTASynthesize} Algorithm, but retaining \textit{similarity
  reduction}; (ii) \textsf{Hegel(-S)}, a variant of \name{} with
support for pruning but \textit{without} similarity reduction (i.e.,
lines 8 and 9 in Algorithm are commented out);
and, (iii) {\sf Hegel(-All)}, a baseline variant that constructs the
LTA without performing any reduction (i.e., removes lines 7-9 in the
algorithm).  We compare these variants in terms of two main metrics,
overall \textit{synthesis times} and the size of the search space in
each case after the reduction, shown by \textit{number of program
  terms enumerated} during search, compared to the base-line ({\sf
  Hegel(-All)}.

\begin{figure}[H]

\centering 
\includegraphics[width=1.0\textwidth]{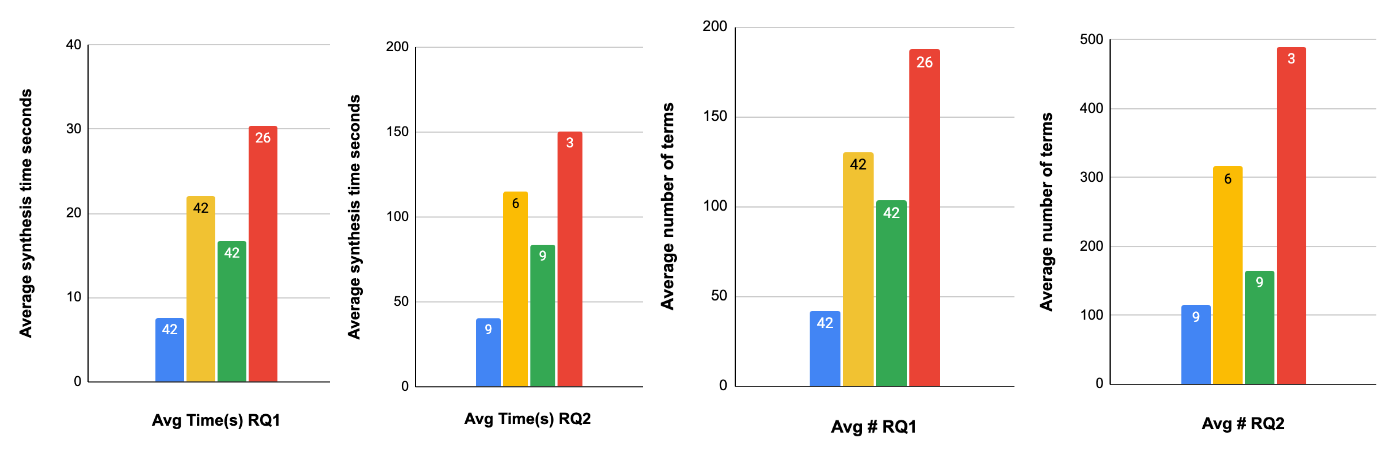}
\caption{Comparison of Hegel and its variants {\color{blue}{\sf
      Hegel}}, {\color{amber}{\sf Hegel(-S)}}, {\color{bargreen}{\sf
      Hegel(-P)}}, {\color{red}{\sf Hegel(-ALL)}}, on average synthesis times
  and the number of candidate terms generated on RQ1 and RQ2. The
  labels on each bar show the number of benchmarks solved by these
  variants out of 42 in the case of RQ1 and 9 in the case of RQ2.}
\label{fig:refined-time-rq1}
\end{figure}

The first two charts in Figure~\ref{fig:refined-time-rq1} show results for overall average
synthesis times across the two sets of benchmark queries described
earlier.  We note that both {\sf Hegel(-S)}, and {\sf Hegel(-P)} can solve
all queries from RQ1, but at a cost which is 2 - 3X greater than
\name{}. {\sf Hegel(-All)} on the other hand fails on almost half of the benchmarks.  
In contrast,
although {\sf Hegel(-P)} was also able to solve the full complement of
queries studied in RQ2, it did so with a considerable larger overhead
compared to \name{}, while here the the irrelevant code pruning ({\sf Hegel(-S)}) alone is insufficient to scale the variant to these challenging benchmarks and it fails to solve 3/9 benchmarks. 
The second pair of charts and show the average number of terms
enumerated by these different variants, showing the reduction of search space by each reduction strategy, with {\sf Hegel(-All)} as the baseline.  Here we see, with the combined reduction strategies, {\sf Hegel} sees the maximum search space reduction, while the other two variants {\sf Hegel(-P)} and {\sf Hegel(-S)} having much larger search spaces, (anywhere from 2-4.5X more) without necessarily solving the same number of queries.

%




\end{document}